\documentclass{article}

\usepackage[letterpaper]{geometry}

\usepackage[normalem]{ulem}

\PassOptionsToPackage{hyphens}{url}
\usepackage[colorlinks=true,citecolor=blue,linkcolor=blue,linktocpage]{hyperref}
\usepackage{makecell}

\usepackage[linesnumbered,ruled,vlined]{algorithm2e}
\SetKwInput{KwInput}{Input}                
\SetKwInput{KwOutput}{Output}
\usepackage{cite}

\usepackage{drawmatrix}
\usepackage{tikz}
\usepackage{xcolor}
\usepackage[export]{adjustbox}
\usetikzlibrary{matrix,positioning}
\usepackage{graphicx}
\usepackage{caption}
\usepackage[labelformat=simple]{subcaption}

\usepackage{epstopdf}
\usepackage{epsfig}
\usepackage{amssymb}
\usepackage{amsmath}
\usepackage{amsthm}
\usepackage{latexsym}
\usepackage{setspace}
\usepackage{bbm}
\usepackage{flushend}
\usepackage{float}
\usepackage{dsfont}
\usepackage{algorithmic,algorithm2e}
\usepackage{xspace}
\usepackage{enumerate}

\usepackage[title]{appendix}
\newcommand\numberthis{\addtocounter{equation}{1}\tag{\theequation}}

\newcommand{\supp}[1]{\mathtt{supp}\left(#1\right)}
\newcommand{\diag}[1]{\mathtt{diag}\left(#1\right)}
\newcommand{\dP}{\mathrm{P}}

\newcommand{\bP}[2]{\dP_{#1}\left(#2\right)}

\newcommand{\bPP}[1]{\dP_{#1}}

\newcommand{\bPr}[1]{{\mathrm{Pr}}\left(#1\right)}

\newcommand{\bE}[2]{{\mathbb{E}}_{#1}\left[{#2}\right]}
\newcommand{\bEE}[1]{{\mathbb{E}}\left[#1\right]}

\newcommand{\cE}{{\mathcal E}}

\newcommand{\cG}{{\mathcal G}}

\newcommand{\cI}{{\mathcal I}}

\newcommand{\cN}{{\mathcal N}}
\newcommand{\mN}{{\mathbbm N}}
\newcommand{\mR}{{\mathbbm R}}

\newcommand{\cS}{{\mathcal S}}

\newcommand{\cU}{{\mathcal U}}

\makeatletter
\newtheorem*{rep@theorem}{\rep@title}
\newcommand{\newreptheorem}[2]{%
\newenvironment{rep#1}[1]{%
 \def\rep@title{#2 \ref{##1}}%
 \begin{rep@theorem}}%
 {\end{rep@theorem}}}
\makeatother

\newtheorem{theorem}{Theorem}

\newtheorem*{corollary*}{Corollary}
\newtheorem{assumption}{Assumption}

\newtheorem*{assumptions*}{Assumptions}
\newtheorem{lemma}[theorem]{Lemma}
\newtheorem*{lemma*}{Lemma}

\newtheorem{lemma-app}{Lemma}[section]

\newreptheorem{theorem}{Theorem}
\newreptheorem{lemma}{Lemma}

\theoremstyle{remark}
\newtheorem{remark}{Remark}
\newtheorem*{remark*}{Remark}
\newtheorem*{remarks*}{Remarks}

\theoremstyle{definition}
\newtheorem{definition}{Definition}

\newcommand{\ed}{\stackrel{{\rm def}}{=}}

\def\undertilde#1{\mathord{\vtop{\ialign{##\crcr
$\hfil\displaystyle{#1}\hfil$\crcr\noalign{\kern1.5pt\nointerlineskip}
$\hfil\tilde{}\hfil$\crcr\noalign{\kern1.5pt}}}}}

\newcommand{\ep}{\varepsilon}

\allowdisplaybreaks 

\newcommand{\nm}{m}  

\newcommand{\dg}{\mu_{0}} 
\newcommand{\on}{\mu_{s}} 
\newcommand{\off}{\mu_{d}} 

\newcommand{\by}{Y}				
\newcommand{\bphi}{\Phi}		

\newcommand{\sbg}{\text{subG}}
\newcommand{\sbx}{\text{subexp}}

\begin{document}

\title{Multiple Support Recovery Using Very Few Measurements Per Sample}

\date{}
\author{{Lekshmi Ramesh} \and
  {Chandra R. Murthy} \and {Himanshu
    Tyagi} }

\maketitle

{\renewcommand{\thefootnote}{}\footnotetext{
\noindent The authors are with the Department of Electrical
Communication Engineering, Indian Institute of Science, Bangalore
560012, India.  Email: \{lekshmi, cmurthy, htyagi\}@iisc.ac.in.} }

		
{\renewcommand{\thefootnote}{}\footnotetext{
		\noindent \thanks{This work was financially supported by a PhD fellowship from the Ministry of Electronics and Information Technology, Govt. of India, and by research grants from the Aerospace Network Research Consortium and the Center for Networked Intelligence (CNI) at the Indian Institute of Science.}} }		
		
{\renewcommand{\thefootnote}{}\footnotetext{
		\noindent \thanks{To appear in ISIT 2021.}} }		

\maketitle

\begin{abstract}
  In the problem of multiple support recovery,
  we are given
    access to linear measurements of multiple sparse samples in $\mR^{d}$. These samples can be partitioned into $\ell$ groups, with samples having the same support belonging to the same group. For a given budget of $m$ measurements per sample, 
the goal is to recover the $\ell$ underlying supports, in the absence of the knowledge of group labels.
We study this problem with a focus on the \emph{measurement-constrained} regime where $m$ is
{smaller} than the support size $k$ of each sample.
We design a two-step procedure that estimates the union of the underlying supports first, and then uses a spectral algorithm to estimate the individual supports.
 Our proposed estimator can recover the supports with $m<k$ measurements per sample,
from $\tilde{O}(k^{4}\ell^{4}/m^{4})$ samples.
Our guarantees hold for a general, generative model assumption on the samples and  measurement matrices. We also provide results from experiments conducted on synthetic data and on the MNIST dataset.
\end{abstract}


\section{Introduction}

We study the problem of {\em multiple support recovery}
using linear measurements, where there are $n$ random samples $X_{1},\ldots,X_{n}$
taking values in
	$\mR^{d}$, such that for each $i\in[n]$,
	$\mathtt{supp}(X_{i})\in\{\cS_{1},\ldots,\cS_{\ell}\}$ {\em almost surely,}\footnote{The support of a vector $x\in\mR^{d}$ is the set $\{u\in[d]:x_{u}\ne 0\}$.} with 
	$\cS_{i}\subset [d]$ and $\cS_{i}\cap\cS_{j}=\emptyset$ for all $i\ne j$. 
	That is, the support of each sample is one out of a small set of $\ell$ allowed supports.
	We assume that the samples $X_{i}$ are sparse and that $|\cS_{i}|=k \ll d$, $i\in[\ell]$.
	We are given low dimensional
	projections of these samples using $m\times d$ matrices
	$\Phi_{1},\ldots,\Phi_{n}$. In our setting, we focus on the regime
        where we have access to very few measurements per sample, namely, when 
	$m<k$. 
        Given access to the projections $Y_{i}=\Phi_{i}X_{i},
        i\in[n]$, and the projection matrices, we seek to recover the underlying supports 
	$\{\cS_{1},\ldots,\cS_{\ell}\}$.

        This is a generalization of the well-studied problem of recovering
        a \emph{single} unknown support from multiple linear measurements \cite{Lounici_COLT_2009,Tang_TIT_2010,Eldar_TIT_2010,Park_TIT_2017,Ramesh_arxiv_2019}, which has been applied to solve inverse problems in imaging, source localization, and anomaly detection \cite{Chen_GRSL_2011,Iordache_TGRS_2014,Malioutov_TSP_2005,Adler_MLSP_2013}.
    It is also related to the study of sparse random effects in mixed linear models \cite{Castro_AnnalsStats_2011,Balasubramanian_UAI_2013}. Mixed linear models are a generalization of linear models where an
    additional additive correction component is included to model a class-specific correction to the average
    behavior. This residual correction term is commonly known as the random effect term. It is often assumed
    to be generated from an unknown prior distribution with zero-mean, coming from a parametric family
    whose parameters are estimated by using the class-specific data. The problem of multiple support recovery is also discussed in \cite{Vaswani_TSP_2016, Mota_TSP_2016} under the assumption of slowly varying supports.
    
    There are two sets of unknowns in the setting described above -- the labels, indicating which support was chosen for each sample, and the $\ell$ supports $\cS_{1},\ldots,\cS_{\ell}$. Note that given the knowledge of the labels, one could group together samples with the same support, and use standard algorithms to recover the support. However, in the absence of labels,  the problem of recovering the supports is much harder.  
        A naive
        scheme could be to just estimate each support individually, which requires $m=O(k\log (d-k))$ measurements per sample \cite{Wainwright_TIT_2009,Aeron_TIT_2010}. But can we do better if we exploit
        the joint structure present across the samples, since there will be several samples that have the same support? In this work, we show that one can operate in the measurement-constrained regime of $m<k$, when a sufficiently large number of samples is available.

  \subsection{Prior work}  
 For the special case with $n=\ell=1$, when there is a single $k$-sparse sample of length $d$, 
 it is known that $m=\Theta(k\log (d-k))$ 
 measurements are necessary and sufficient to recover the support
 \cite{Wainwright_TIT_2009} with noisy measurements, when the inputs are worst-case. For the case
 with a single common support across multiple samples (i.e., $\ell=1$ and $n>1$), several previous works have studied the question of support recovery in the $m>k$ setting \cite{Tang_TIT_2010,Eldar_TIT_2010,Park_TIT_2017}.
 
On the other hand, in the $m<k$ regime, it was 
 shown recently in~\cite{Ramesh_ISIT_2019,Ramesh_arxiv_2019} that
 $n=\Theta((k^{2}/m^{2})\log d)$ samples are necessary and
 sufficient, 
 assuming
 a subgaussian generative model on the samples and
 measurement matrices and that the measurement matrices are drawn independently
 across samples.
 In fact, the lower bound of~\cite{Ramesh_arxiv_2019} applies to the worst-case
   setting as well, showing that while $k$ overall measurements\footnote{The overall measurements
     in our model are $nm$.} suffice when $m$ exceeds $k$, at least (roughly) $k^2/m$ measurements
   are required when $m<k$.

 In \cite{Obozinski_annals_stats_2011}, the problem of recovering the
 union of supports from linear measurements is considered. The setting
 allows for overlaps in the supports, but otherwise places no
 constraints.  The results when applied to the case of disjoint
 supports lead to a requirement of $m=O(k\log d)$ measurements per
 sample, and therefore are not applicable to our setting.
 Another line of related works is on multi-task learning/multi-task
 sparse estimation \cite{Wang_UAI_2015,Qi_ICML_2008,Argyriou_NIPS_2006} that use hierarchical Bayesian models and
 focus on recovering the samples, rather than the supports,
   and so still require at least $k$ measurements per sample.  However, none of these results shed light on how to recover multiple
 supports when we are constrained to observe less than $k$
 measurements per sample.
 
We note that there has been some recent work in the literature on
mixture of sparse linear regressions that considers the related
problem of recovering multiple sparse vectors from linear measurements
\cite{Yin_TIT_2019,Krishnamurthy_NIPS_2019,Li_COLT_2018,Chen2_arxiv_2019,Argyriou_NIPS_2006,Obozinski_SC_2010}.  The model shares some similarities with the
$m=1$ case in our setting, but there are some important differences.
Unlike our setting, these works consider the samples to be
deterministic and do a worst-case analysis.  Further, when $\ell=1$ in
the mixture of sparse linear regressions setting, we have multiple observations from the
{same} unknown sparse vector, thus reducing the problem to the
standard compressed sensing problem.  On the other hand, with
$\ell=m=1$ in our setting, we obtain a single observation from
{different} sparse vectors sharing a common support. The latter
setting is harder and requires $\Omega(k^{2}\log d)$ samples to recover the common support~\cite{Ramesh_arxiv_2019}.

 \subsection{Contributions and Techniques}
 Our approach builds on the following
 simple but crucial observation: since each sample is $k$-sparse with
 support equal to one of the $\cS_{i}$ (with the $\cS_{i}$ being
 disjoint), the sample covariance matrix
 $(1/n)\sum_{i=1}^{n}X_{i}X_{i}^{\top}$ exhibits a block structure
 under an unknown permutation of rows and columns.  This motivates the
 use of spectral clustering to recover the underlying
 supports. However, we only have access to low-dimensional projections of the
 data. 
To circumvent this difficulty, we compute
 $\Phi_{i}^{\top}Y_{i}$ and use these as a proxy for the data, and form an estimate of the diagonal entries of the covariance matrix of the samples.  We
 build further on this idea and propose an estimator that first determines the
 union of the $\ell$ supports from $\Phi_{i}^{\top}Y_{i}$ using
 the estimator in \cite{Ramesh_arxiv_2019}.
                We then construct an affinity matrix using  the proxy samples
$\Phi_{i}^{\top}Y_{i}$
and apply spectral clustering to estimate individual supports from the union.

                This clustering based approach to support recovery is new, and very different from traditional approaches to sparse recovery in the multiple sample setting. It reduces the support recovery problem to that of recovering the structure of a certain block matrix, a question which has been studied in the literature on community detection on graphs \cite{McSherry_FOCS_2001,Newman_PHYRE_2006,Hajek_JMLR_2016,Abbe_JMLR_2017}, and for which many algorithms are known. However, unlike the community detection problem where an instance of the adjacency matrix is available as an observation, the affinity matrix constructed in our case has a more complicated structure and requires a separate, careful analysis.

        We show that using our algorithm, it is possible to recover
                all the supports with \emph{fewer} than $k$ measurements per
                sample. Our algorithm is easy to implement and has computational complexity that scales linearly with ambient dimension $d$ and number of samples $n$.
                Our main result is an upper bound on the sample complexity of the multiple support recovery problem, stated in Theorem \ref{thm_multiple_supp_disj}.  In similar spirit to \cite{Ramesh_arxiv_2019}, which studied the case of a single unknown support in the measurement-constrained regime of $m<k$, our work provides an algorithm for the multiple support recovery problem in this regime. 
The analysis of our algorithm involves studying spectral properties of the (random) affinity matrix that has dependent and heavy-tailed entries.  
  We characterize these spectral quantities for the expected affinity matrix, which we show has a block structure, and then use results from matrix perturbation and matrix concentration to obtain performance guarantees for our algorithm.

Also, we provide experimental results on synthetic and real datasets, and show that the proposed algorithm is able to recover the unknown supports with very few measurements per sample. While our guarantees are for the case of disjoint supports, some simple heuristics can be used to handle the case of overlapping supports in practice, as we show in Section \ref{sec:simulation}.

\subsection{Organization}
In the next section, we formally state the problem and the assumptions
we make in our generative model setting. This is followed by a statement of our main
result, which provides an upper bound on the sample complexity of multiple support
recovery. 
We describe the estimator in Section \ref{sec:estimator}, and analyze its performance in Section \ref{sec:analysis}.
We provide experimental results in Section \ref{sec:simulation}.
The technical results required for the proofs in Section \ref{sec:analysis} are available in the appendices.
\subsection{Notation}  
For a matrix $A$, we denote its $(u,v)$th entry by $A_{uv}$. For a collection of matrices $\{A_{i}\}_{i=1}^{n}$, we use $A_{i}(u,v)$ to denote the $(u,v)$th entry of the $i$th matrix.
Also, for a vector $X_{j}$, $X_{ji}$ denotes the $i$th component of $X_{j}$. For sets $\cS$ and $\cS^{\prime}$, $\cS\Delta\cS^{\prime}=(\cS\backslash\cS^{\prime})\cup(\cS^{\prime}\backslash\cS)$ denotes their symmetric difference.
For a vector $a\in\mR^{d}$, $\supp{a}$ denotes the subset $\{i\in[d]:a_{i}\ne 0\}$, $\mathtt{diag}(a)$ denotes the $d\times d$ diagonal matrix with entries of $a$ on the diagonal, and $[d]$ denotes the set $\{1, 2, \ldots, d\}$.
For a matrix $A$, we use $\|A\|_{op}\ed\sup_{\|x\|_{2}=1}\|Ax\|_{2}$ to denote the operator norm of $A$. When $A$ is symmetric, $\|A\|_{op}$ equals the magnitude of the largest eigenvalue of $A$.
We use the shorthand $Z_{1}^{n}$ to denote independent and identically distributed random variables $Z_{1},\ldots,Z_{n}$. For $u>0$, we use $\Gamma(u)\ed\int_{0}^{\infty}x^{u-1}e^{-x}dx$ to denote the gamma function.

\section{Problem formulation and main result}\label{sec:formulation}

We consider a Bayesian setup for modeling samples $X_1 , \ldots, X_n$ taking values in $\mR^d$ with
$\supp{X_i}\ed \{j\in [d]: X_{ij}\neq 0\}\in \{\cS_1,\ldots,\cS_\ell\}$,
where $\cS_i \subset [d]$ are unknown sets such that $|\cS_i|=k$. Specifically,
we consider distributions 
$\dP^{(1)}, \ldots, \dP^{(\ell)}$ with\footnote{We consider distributions $\dP$ with densities $f_\dP$ with respect to the Lebesgue measure 
	and define $\supp{\dP}=\{x\in \mR^d: f_\dP(x)>0\}$.} 
\[
\supp{\dP^{(i)}}=\{x\in\mR^{d}:\mathtt{supp}(x)=\cS_{i}\}, \quad i \in [\ell],
\]
and $n$ i.i.d. samples $X_{1},\ldots,X_{n}$ taking values in $\mR^{d}$
and generated from a common mixture distribution 
\begin{align}
\label{mixture}
  \dP_{\cS_{1},\ldots,\cS_{\ell}}=\frac 1 \ell \sum_{i=1}^{\ell} \dP^{(i)},
\end{align}
parameterized by the tuple $(\cS_{1}\ldots,\cS_{\ell})$.
In fact, we assume that
$\dP^{(i)}$ is a multivariate subgaussian distribution (see Appendix~\ref{app:moment_concentration} for the definition of a subgaussian random variable) with zero mean and diagonal covariance matrix $K_{\lambda_{i}}=\diag{\lambda_{i}}$,
where the parameter $\lambda_{i}$ is a $d$-dimensional vector for which $\supp{\lambda_i}=\cS_i$, $i\in[\ell]$.  More concretely, we make
the following assumption. 
\begin{assumption}\label{assump_x}
	For a sample $X_{j}\sim \dP^{(i)}$, $j\in[n]$, $i \in [\ell]$, and an absolute constant $c$, 
	$\bE{\dP^{(i)}}{X_{j}X_{j}^T}=\diag{\lambda_i}$ with $\lambda_i \in \mR_+^d$, $\supp{\lambda_i}=\cS_i$, and  $X_{j}$ has independent, zero mean entries with its $t$th entry $X_{jt}$ satisfying
	$X_{jt} \sim\sbg(c \lambda_{it})$, $t\in[d]$. 
	{Furthermore, for each $i\in [\ell]$ and $t\in\cS_{i}$,  $\lambda_{it}=\lambda_{0}>0$, and $\bE{\dP^{(i)}}{X_{jt}^{4}}=\rho$.}
\end{assumption}

For samples $X_1, \ldots, X_n$ generated as above, we are given access to
projections $Y_{i}=\Phi_{i}X_{i}$, $i\in [n]$, where the matrices
$\Phi_{i}\in\mR^{m\times d}$ are random and independent for different $i\in [n]$.
Our analysis requires handling higher order moments of the entries of the measurement matrices, which motivates the following assumption.

\begin{assumption}\label{assump_phi}
	The {$m\times d$} measurement matrices
        $\Phi_{1},\ldots,\Phi_{n}$ are independent, with entries
        that are independent and zero-mean. Furthermore,
                $\Phi_{i}(u,v)\sim\sbg(c^\prime/m)$, and the moment
                conditions $\bEE{\Phi_{i}(u,v)^{2}}=1/m$ and $\bEE{\Phi_{i}(u,v)^{2q}}=c_{q}/m^{q}$ hold for
                $q\in\{2,3,4\}$, where $c_{q}$ and $c^{\prime}$ are absolute
                constants.
\end{assumption}
 The assumption above holds, for example, when $\Phi_{i}(u,v)\sim\cN(0,1/m)$ or when $\Phi_{i}(u,v)$ are Rademacher, i.e., take values from $\{1/\sqrt{m},-1/\sqrt{m}\}$ with equal probability. 
Also, these moment assumptions can be relaxed to hold up to constant factors from above and below, i.e.,   $\bEE{\Phi_{i}(u,v)^{2q}}=\Theta(1/m^{q})$.

Our goal is to recover the supports $\{\cS_{1},\ldots,\cS_{\ell}\}$ using
 $\{Y_{i},\Phi_{i}\}_{i=1}^{n}$. The error criterion will be the average of the per support errors, measured using the set difference between the true and estimated supports. 
Specifically, denote by $\Sigma_{\ell,d}^{\prime}$ the set consisting of all $\ell$ tuples of subsets
$(\cS_1, \ldots, \cS_\ell)$ such that $\cS_{i}\subset [d]$, $i\in[\ell]$, and $\cS_{i}\cap\cS_{j}=\emptyset$, for all $i\ne j$.
Let $\Sigma_{k,\ell,d}\subset\Sigma_{\ell,d}^{\prime}$ be such that $|\cS_i|=k$, for all $i\in[\ell]$. Denote by $\cG_{\ell}\ed \{\sigma:[\ell]\rightarrow [\ell]\}$ the set of all permutations on $[\ell]$. We have the following definition.
\begin{definition}
An {\em $(n,\varepsilon,\delta)$-estimator for $\Sigma_{k,\ell,d}$} is a mapping
$e: (Y_{1}^{n},\Phi_{1}^{n})\mapsto (\hat{\cS}_1,\ldots, \hat{\cS}_\ell)\in \Sigma_{\ell,d}^{\prime}$
for which
\begin{align}
  \dP_{\cS_{1},\ldots,\cS_{\ell}}\bigg(\exists\,\sigma\in\cG_{\ell}~\text{s.t.}~\sum_{i=1}^\ell\left|\cS_{i}\Delta\hat{\cS}_{\sigma(i)}\right|< k\varepsilon\ell^{2}\bigg)
 \geq 1-\delta, 
  \label{err_eps}
\end{align}
for all $(\cS_1, \ldots, \cS_\ell)\in\Sigma_{k,\ell,d}$, where $\cS_{1}\Delta \cS_{2}$ denotes the symmetric difference between
sets $\cS_{1}$ and $\cS_{2}$.
\end{definition}
For fixed $\ell, m, k, d,
\varepsilon$, and $\delta$, the least $n$ such that we can find an
$(n,\varepsilon, \delta)$-estimator for $\Sigma_{k,\ell,d}$ is termed
the {\em sample complexity of multiple support recovery}, which we denote by $n^{*}(\ell,m,k,d,\varepsilon,\delta)$. In our main result stated below, we provide an upper bound on $n^{*}(\ell,m,k,d,\varepsilon,\delta)$.

\begin{theorem}\label{thm_multiple_supp_disj}
	Let $m,k, d, \ell\in \mN$ with
        $\log k\ge 2$. Further, let $(\log k\ell)^{2}\le m< k$, and  $1/k\ell \leq \varepsilon\leq 1/\ell$. Then, under
        Assumptions \ref{assump_x} and \ref{assump_phi}, the sample complexity of multiple support recovery satisfies        
\begin{align*}
  n^{*}(\ell,m,k,d,\varepsilon,\delta)
  =O\!\left(\!\max\bigg\{\frac{1}{\varepsilon}\bigg(\frac{k\ell}m\right)^{4} (\log k)^{4}\log k\ell\log\frac{1}{\delta},
\frac{k^{2}\ell^{2}}{m^{2}}\log \frac{k\ell(d-k\ell)}{\delta}\bigg\}\bigg).
\end{align*}      
\end{theorem}
  
   \begin{remark}
    	For values of $\varepsilon$ lower than $1/\ell k$, the result from Theorem \ref{thm_multiple_supp_disj} continues to hold with $\varepsilon$ set to $1/\ell k$. This is because $\varepsilon=1/\ell k$ corresponds to exact recovery of the supports.
    \end{remark}
  We present the
algorithm that attains this performance in the next section, and prove the theorem in Section \ref{proof_thm_multiple_supp_disj}.

Our estimator works in two steps by estimating the union of supports first and then estimating each support, and the sample complexity bound above is obtained by analyzing each of the two steps. To the best of our knowledge, this is the first estimator that can recover multiple supports under the constraint of $m<k$ linear measurements per sample.
We also note that for the problem of recovering a single support exactly,
it was shown in~\cite{Ramesh_arxiv_2019} that roughly
$\Omega((k/m)^2\log k (d-k))$ samples are necessary. Thus, our sample
complexity upper bound above matches this lower bound
quadratically. However, there is a gap between the lower bound and the
upper bound, which is an interesting problem for future research.

\section{The estimator}\label{sec:estimator}
Our first step will be to recover the union of the $\ell$ underlying
supports, and then refine this estimate to finally recover the
individual supports. To estimate the union, we use the estimator
described in \cite{Ramesh_ISIT_2019}. Following this, we use a
spectral clustering based approach to recover the individual
supports. We provide more details in the next two subsections.
\subsection{Recovering the union of supports}\label{sec:estimator_union}
We first observe that 
 the samples $X_i$ have an effective covariance matrix whose diagonal 
  has support equal to the union of the supports, which allows us to use
  the results from \cite{Ramesh_arxiv_2019} to recover the union. Specifically, 
we form ``proxy samples'' 
$\hat{X}_{i}=\Phi_{i}^{\top}Y_{i}=\Phi_{i}^{\top}\Phi_iX_i$
and use the diagonal of the sample covariance matrix of $\hat{X}_{i}$ as an estimate
for the diagonal of the covariance matrix for $X_i$. We will show that the $k\ell$ largest
entries of the recovered diagonal correspond to the union of the
supports.

Formally, define $\cS_{\text{un}}\ed\cup_{i=1}^{\ell}\cS_{i}$ to be
the union of the $\ell$ unknown disjoint supports and note that
$|\cS_{\text{un}}|=k\ell$.  We use the estimator described in~\cite{Ramesh_arxiv_2019} and form the statistic
$\tilde{\lambda}\in\mR^{d}$ as follows. First, define vectors
$a_{1}^{\prime},\ldots,a_{n}^{\prime}$ with entries
\begin{align}\label{eq:def_a} 
a_{ji}^{\prime}\ed(\Phi_{ji}^{\top}\by_{j})^{2},\quad i\in[d].
\end{align}
Each $a_{j}^{\prime}$, $j\in[n]$, can be thought of as a crude estimate for the variances along the $d$ coordinates obtained using the $j$th sample. We then define the average of these vectors as
\begin{align}\label{eq:def_lambdatil}
\tilde{\lambda}\ed\frac{1}{n}\sum_{j=1}^{n}a_{j}^{\prime}.
\end{align}
This statistic captures the variance along each coordinate of $X_i$.
Due to the
averaging across samples, we expect a larger value of the
statistic along coordinates that are present in at least one of the supports. On the
other hand, coordinates that are not present any support should
result in a smaller value of the statistic.
As shown in \cite{Ramesh_arxiv_2019}, such a separation between the estimate values indeed occurs when $n$ is sufficiently large.
The algorithm declares the indices of the $k\ell$ largest entries of $\tilde{\lambda}$ as the estimate for $\cS_{\text{un}}$. 
Letting $\tilde{\lambda}_{(1)}\ge\cdots\ge\tilde{\lambda}_{(k\ell)}$
represent the sorted entries of $\tilde{\lambda}$, the estimate
$\hat{\cS}_{\text{un}}$ for the union is 
\begin{align}\label{est_union}
\hat{\cS}_{\text{un}}=\{(1),\ldots,(k\ell)\},
\end{align}
where we assume the size of the union to be known. In practice, $\tilde{\lambda}$ can be used to estimate the size of the union as well
by
sorting the entries of $\tilde{\lambda}$ and using the index where there is a sharp decrease in the values as the estimate for $k\ell$, similar to the approach of using scree plots to determine model order in problems such as PCA \cite{Zhu_Elsevier_2006}. 


\subsection{Recovering individual supports}\label{sec:estimator_cluster}
We now describe the main step of our algorithm where we partition the coordinates in $\hat{\cS}_{\text{un}}$
recovered in the first step into disjoint support estimates $\hat{\cS}_{1},\ldots,\hat{\cS}_{\ell}$.
We will use $a_{1}^{\prime},\ldots,a_{n}^{\prime}$ described in \eqref{eq:def_a} for this purpose.
Since we now have an estimate for the union, we will restrict $a_i^\prime$ to coordinates in
$\hat{\cS}_{\text{un}}$, and denote them as  $a_i\in\mR^{k\ell}_+$. Also, without loss of generality, we set
$\hat{\cS}_{\text{un}}=[k\ell]$.\footnote{This is to keep notation simple. For a general
  $\hat\cS_{\text{un}}$, we can have a function
  $g:[k\ell]\rightarrow \hat{\cS}_{\text{un}}$ that provides the
  mapping of each coordinate of $a_{i}$ to its corresponding value in
  $\hat{\cS}_{\text{un}}$ as indicated in step $7$ of Algorithm \ref{algo_mutliple_disj}.}
  \usetikzlibrary{matrix}
  \usetikzlibrary{calc,fit}
  \tikzset{%
  	highlight1/.style={rectangle,rounded corners,color=red!,fill=blue!15,draw,fill opacity=0.5,thick,inner sep=0pt}
  }
  \tikzset{%
  	highlight2/.style={rectangle,rounded corners,color=red!,fill=blue!15,draw,fill opacity=0.5,thick,inner sep=0pt}
  }
  \begin{figure}[t]
  	\begin{equation*}
  	\renewcommand{\arraystretch}{1.5}
  	\bEE{T}=
  	\begin{tikzpicture}[baseline=(m.center)]
  	\matrix (m) [matrix of math nodes, left delimiter={[}, right delimiter={]},
  	row sep=1mm, nodes={minimum width=3em, minimum height=1.6em}] {
  		\mu_{0} & \mu^{s} & \mu^{d}  & \mu^{d} \\
  		\mu^{s} &  \mu_{0} & \mu^{d} & \mu^{d}  \\
  		\mu^{d}& \mu^{d} & \mu_{0} & \mu^{s} \\
  		\mu^{d}& \mu^{d} & \mu^{s} & \mu_{0} \\
  	}; 
  	\node[highlight2, fit=(m-1-1.north west) (m-2-2.south east)] {};
  	\node[highlight1, fit=(m-3-3.north west) (m-4-4.south east)] {};
  	\end{tikzpicture}
  	\adjustbox{raise=.7cm}{$\Bigg\}\cS_{1}$}	
  	\hspace{-.6cm}\adjustbox{raise=-.9cm}{$\Bigg\}\cS_{2}$}
  	\end{equation*}
  	\caption{Block structure of the expected clustering matrix when $\ell=2$ and the supports are disjoint, under appropriate permutation of rows and columns. }\label{fig_block_mtx}
  \end{figure}
  
  Our approach is the following: we construct a $k\ell\times k\ell$ \emph{affinity matrix} $T$ and perform spectral clustering using this matrix, which will partition the coordinates in $[k\ell]$ into $\ell$ groups.
  The main step here is to construct  an affinity matrix $T$ that can provide reliable clustering, and we will use the per-sample variance estimates $a_{1},\ldots,a_{n}$ for this purpose. 
  The idea is that for any coordinate pair $(u,v)\in[k\ell]\times [k\ell]$, if both $u$ and $v$ belong to the same support, then we expect the product $a_{iu}a_{iv}$ to have a ``large" value for most of the sample indices $i\in[n]$. On the other hand, if $u$ and $v$ belong to different supports, then $a_{iu}a_{iv}$ will be close to zero for most $i\in[n]$. Although each $a_{i}$ individually is not a good estimate for the support of $X_{i}$, the averaging over $n$ makes the estimate reliable. Formally, we construct the $k\ell \times k \ell$ matrix $T$ with entries
  \begin{equation}\label{eq:T_mtx_entry}
  T_{uv}\ed\frac{1}{n}\sum_{j=1}^{n}a_{ju}a_{jv},\quad (u,v)\in[k\ell]\times[k\ell].
  \end{equation}
  The key observation here is that the \emph{expected} value of the random matrix $T$ has a block structure when the rows and columns are appropriately permuted, and this block structure corresponds to memberships of each of the indices in $[k\ell]$ to one of the underlying supports. This is illustrated in Figure \ref{fig_block_mtx} for $\ell=2$, and we will examine this structure in detail in the next section. 
  A well-known method to find these memberships is to use spectral clustering \cite{Rohe_annals_stats_2011,Newman_PHYRE_2006}, which uses properties of the eigenvectors of block-structured matrices to determine the partition. For instance, when $\ell=2$, the \emph{sign} of the second leading eigenvector of $\bEE{T}$ provides a way to partition the coordinates in $[k\ell]$ into two groups. When $\ell>2$, spectral clustering makes use of multiple eigenvectors and a nearest neighbor step to identify the partition. A full description of the solution in the general case is provided in Algorithm \ref{algo_mutliple_disj}.
  


  In practice, we only have access to $T$,
and not $\bEE{T}$ to which the discussion above applies.  In
what follows, we show that the eigenvectors of $T$ itself suffice,
provided we have sufficiently many samples. At a high level, our
analysis follows that of spectral clustering in the stochastic block model (SBM)
setting and the goal is to show that the eigenvectors of $\bEE{T}$ and
its ``perturbed'' version $T$ are close to each other. This can be shown
using the Davis-Kahan theorem from matrix perturbation theory, which states that the angle between any two
corresponding eigenvectors of $T$ and $\bEE{T}$ is small provided the
error matrix $T-\bEE{T}$ has small spectral norm.
The key challenge, therefore, is to control $\|T-\bEE{T}\|_{op}$. 

  Unlike typical settings, the entries of $T$ are not independent, in addition to being heavy tailed. Standard methods based on the $\varepsilon$-net argument are, therefore, difficult to apply in this setting.
  One strategy could be to show exponential concentration around the mean for
  \emph{each} entry of $T$. Once each entry of $T$ is bounded with high probability, one can bound the Frobenius
  norm and therefore the spectral norm of the error matrix. 
  However, the moment generating function (MGF) of 
  each summand in \eqref{eq:T_mtx_entry} is unbounded, so 
  deriving a tail bound for the sum requires 
  a more careful tail splitting method (see, for example, \cite[Exercise 2.1.7]{Tao_RMT_2016}), and leads to measurement matrix dependent quantities that are difficult to handle. Due to the same reason, techniques from matrix concentration that involve bounding the MGF of the summands \cite[Theorem 6.1, Theorem 6.2]{Tropp_FOCM_2012} cannot be used in our setting.
  
  To circumvent this difficulty, we turn to a beautiful result by Rudelson \cite{Rudelson_JFA_1999}, that characterizes the
  expected value of the quantity $\|T-\bEE{T}\|_{op}$, when $T$ is a sum of independent rank-one matrices and only requires certain moment assumptions on the summands. This is exactly our setting since \eqref{eq:T_mtx_entry} can equivalently be represented as $T=(1/n)\sum_{i=1}^{n}a_{i}a_{i}^{\top}$. An
  application of Markov inequality followed by the Davis-Kahan theorem then shows that the eigenvectors of
  $T$ and $\bEE{T}$ are close to each other. We provide more details about the analysis in the next section. 

    \begin{algorithm}[th]
    	\DontPrintSemicolon	
    	\KwInput{Measurements $\{Y_{i}\}_{i=1}^{n}$, Measurement matrices $\{\Phi_{i}\}_{i=1}^{n}$, $k$, $\ell$}
    	\KwOutput{Support estimates $\hat{\cS}_{1},\ldots,\hat{\cS}_{\ell}$}
    	
    	Form variance estimates $a_{1}^{\prime},\ldots,a_{n}^{\prime}$ with entries
    	\begin{align*}
    	a_{ji}^{\prime}=(\Phi_{ji}^{\top}Y_{j})^{2},\quad i\in[d].
    	\end{align*}	
    	
    	Compute
    	\begin{align*}
    	\tilde{\lambda}=\frac{1}{n}\sum_{i=1}^{n}a_{i}^{\prime}.
    	\end{align*}
    	Sort entries of $\tilde{\lambda}$ to get $\tilde{\lambda}_{(1)}\ge\cdots\ge\tilde{\lambda}_{(d)}$ and output estimate for union
    	\begin{align*}
    	\hat{\cS}_{\text{un}}=\{(1),\ldots,(k\ell)\}.
    	\end{align*}
    	
    	Restrict $a_{1}^{\prime},\ldots,a_{n}^{\prime}$ to the coordinates in $\hat{\cS}_{\text{un}}$, to get $a_{1},\ldots,a_{n}$.
          Also, let $g:[k\ell]\rightarrow \hat{\cS}_{\text{un}}$ denote the mapping from the coordinates of $a_{i}$ to the true coordinate in $\hat{\cS}_{\text{un}}$.
    	
    	 Construct affinity matrix $T\in\mR^{k\ell\times k\ell}$ as
    	\begin{align*}
    	T=\frac{1}{n}\sum_{i=1}^{n}a_{i}a_{i}^{\top}.
    	\end{align*}

    	Compute the $\ell$ leading eigenvectors $\hat{v}_{1},\ldots,\hat{v}_{\ell}$
  of $T$
          and let these be the columns of $\hat{V}\in\mR^{k\ell\times l}$. 
    	
    %
    	
    	{\em (The $\ell$-means step)} Find
    	        $C=\arg\min_{U\in \cU_\ell}\|U- \hat{V}\|_F^2$, where
    	        $\cU_\ell$
    	        is the set of all $k\ell\times \ell$ matrices with at most $\ell$
    	        distinct rows.
    	 
    		Denote the indices of identical rows of $C$ as 
    	        sets $\hat{\cS}_1^{\prime}, \ldots, \hat{\cS}_\ell^{\prime}$. 
    	        	Declare
    	        	\begin{align*}
    	        	\hat{\cS}_{i}=\{g(j)\in\hat{\cS}_{\text{un}}: j\in\hat{\cS}_{i}^{\prime}\}.
    	        	\end{align*}
    	\caption{Multiple support recovery}
    	\label{algo_mutliple_disj}
    \end{algorithm}
  
  \section{Analysis of the estimator}\label{sec:analysis}
  \subsection{Recovering the union: Analysis}\label{sec:analysis_union}
  Our analysis of the
  probability of exactly recovering $\cS_{\text{un}}$ using the
  estimator in \eqref{est_union} follows the approach in
  \cite{Ramesh_arxiv_2019}. The key difference is that the samples are
  now drawn from a \emph{mixture} of subgaussian distributions.
  In the next result, we show that if $X$ is drawn from the mixture described in \eqref{mixture}, then
  it is subgaussian with covariance matrix $K_{\lambda_{\text{un}}}$
  where $\lambda_{\text{un}}=\lambda_{1}\lor\cdots\lor \lambda_{\ell}$,
  where $\lor$ denotes entrywise maximum. This helps us to determine the effective parameter that characterizes the input distribution, after which we can use the result from \cite{Ramesh_arxiv_2019}. We prove this result for the two component mixture; it can be extended easily to the general case.
  \begin{lemma}\label{lem:mix_subg}
  	Let $X$ and $Y$ be zero-mean subgaussian random variables with parameters $a^{2}$ and $b^{2}$, respectively. 
  	Further, let $\dP_{X}$ and $\dP_{Y}$ denote the distributions
          of $X$ and $Y$. Then, the random variable $Z$ with distribution
          given by the mixture $q\dP_{X}+(1-q)\dP_{Y}$ with $q\in[0,1]$
          is subgaussian with parameter $\max\{a^{2},b^{2}\}$.  
  \end{lemma}
  \begin{proof}
  Upon bounding the MGF of $Z$, we see that
  	\begin{align*}
  	\bEE{e^{\theta Z}}&=q\bEE{e^{\theta X}}+(1-q)\bEE{e^{\theta Y}}\\
  	&\le q e^{\frac{\theta^{2} a^{2}}{2}}+(1-q)e^{\frac{\theta^{2} b^{2}}{2}}\\
  	&\le e^{\frac{\theta^{2} c^{2}}{2}},
  	\end{align*}
  	where $c=\max\{a,b\}$.
  \end{proof}  
  Thus, the samples $X_{1}, X_2, \ldots, X_n$
  have entries that are independent and subgaussian with covariance
  matrix $K_{\lambda_{\text{un}}}$, where   
  $\lambda_{\text{un}}=\lambda_{1}\lor\cdots\lor \lambda_{\ell}$.
  Therefore, results from \cite{Ramesh_arxiv_2019} imply that we can
  recover $\cS_{\text{un}}$ from the variance estimate
  \eqref{eq:def_lambdatil} by retaining the $k\ell$ largest entries. In
  particular, a direct application of 
  \cite[Theorem 3]{Ramesh_arxiv_2019} with support size set to $k\ell$, gives us
  the following result. 
  \begin{theorem}\label{thm:supp_union}
  	Let $\hat{\cS}_{\text{un}}$ described in \eqref{est_union} be the estimate for the union $\cS_{\text{un}}$. Then, for every $\delta>0$,
  	\begin{align*}
  	\bPr{\hat{\cS}_{\text{un}}\ne \cS_{\text{un}}}\le\delta,
  	\end{align*}	
  	provided $m\ge (\log k\ell)^{2}>1$, and
  	\begin{align*}
  	n\ge  c~\bigg(\frac{k^{2}\ell^{2}}{m^{2}}\log\frac{k\ell(d-k\ell)}{\delta}\bigg),
  	\end{align*}
  	for an absolute constant $c$.
  \end{theorem}
  
 As we discussed in the introduction, if we had labels for each sample indicating which support it belongs to, we could directly use the estimator from \cite{Ramesh_arxiv_2019} after grouping the samples with the same support together. This would require $O((k^{2}\ell/m^{2})\log k (d-k))$ samples. On the other hand, when the labels are unknown, the number of samples required even to estimate the union of the supports is higher, as seen from the theorem above.
  \subsection{Recovering individual supports: Analysis}\label{sec:analysis_cluster}
  Our analysis is based on the fact that the expected affinity matrix has a block structure (under an appropriate permutation of its rows and columns), which we prove in the next lemma. 
  \begin{lemma}[Block structure of $\bEE{T}$]\label{lem:ET_structure}
  Under Assumptions \ref{assump_x} and \ref{assump_phi}, for the matrix $T\in\mR^{k\ell\times k\ell}$
  in \eqref{eq:T_mtx_entry}, $\bEE{T}$ has entries
  given by
  		\begin{align*}
  		\bEE{T_{uv}}
  		=\begin{cases}
  		\dg,~\text{if}~u=v,\\
  		\on,~\text{if}~u\ne v, (u,v)\in\cS_{i}\times \cS_{i}~\text{for any}~i\in[\ell],\\
  		\off,~\text{otherwise},
  		\end{cases}
  		\end{align*} 
  		where the parameters $\dg$, $\on$, and $\off$ depend on $k$, $m$, and $\ell$ and can be explicitly calculated.
  \end{lemma}	
  
  The proof of Lemma \ref{lem:ET_structure} appears in Appendix \ref{app:ET_structure} and involves computing the expected values of expressions containing higher order terms in $\Phi_{i}$ and $X_{i}$.
  Before we proceed, we note the following extension of the ``median trick'' (see, for example, \cite{Chakrabarti_notes_2020}) which shows that
  the dependence of sample complexity on $\delta$ is at most a factor of $O(\log 1/\delta)$,
  provided we can find an $(n,\varepsilon,1/4)$-estimator. 
  
  \begin{lemma}[Probability of error boosting]
  	\label{lem:median_trick}
  	For $\delta \in (0,1)$ and $\ell\in \mN$, if we can find an $(n,\varepsilon,1/4)$-estimator for $\Sigma_{k,\ell,d}$,
  	then we can find an $\left(n\lceil 8 \log \frac 1 \delta\rceil,3\varepsilon, \delta\right)$-estimator for $\Sigma_{k,\ell,d}$.
  \end{lemma}
  We provide the proof in Appendix~\ref{app:median_trick}. 
  
  Thus, from here on, we fix our error requirement to $\delta=1/4$ and seek $(n,\varepsilon, 1/4)$-estimators with the least possible $n$.
  We characterize the performance of the clustering step in the following theorem. The analysis of this step is conditioned on exact recovery of the union ${\cS}_{\text{un}}$ in the first step.
  
  \begin{theorem}\label{thm_ecc}
  	Let $\nu_{1}\ge\cdots\ge\nu_{k\ell}$ denote the ordered
          eigenvalues of $\bEE{T}\in\mR^{k\ell\times k\ell}$, and
          define 
  	$\Delta_{\ell}=\nu_{\ell}-\nu_{\ell+1}$ when $\ell\geq 2$.
  	For every $\varepsilon \in [1/\ell k, 1/\ell)$, we can find an $(n,\varepsilon,1/4)$-estimator for $\Sigma_{k,\ell,k\ell}$ provided
  	  \[n\ge 	 c\frac{\max\{1,\|\bEE{T}\|_{op}\}}{\varepsilon\Delta_{\ell}^{2}}\cdot
                    {\bEE{\max_{i\in[n]}\|a_{i}\|_{2}^{2}}}\cdot \log k\ell ,\]
for an absolute constant $c$.         
  \end{theorem}	
  
  The result above applies to any setting where we have i.i.d. samples $a_{1},\ldots,a_{n}$ whose covariance has a block structure under permutation, and the goal is to group the coordinates of $a_{i}$ based on the unknown block structure. 
  We provide the proof of Theorem \ref{thm_ecc} at the end of this section.
  
  The next two results provide us with bounds on the spectral quantities $\|\bEE{T}\|_{op}$ and $\Delta_{\ell}$, and on $\bEE{\max_{i\in[n]}\|a_{i}\|_{2}^{2}}$ appearing in Theorem \ref{thm_ecc}.
  
  \begin{lemma}\label{lem:ET_spectrum}
  Under Assumptions \ref{assump_x} and \ref{assump_phi},
  	we have
  	\begin{align*}
  	\|\bEE{T}\|_{op}\le
  	\rho\frac{k^{2}\ell}{m^{2}}+\lambda_{0}^{2}\frac{k^{3}\ell}{m^{2}},
  	~\text{and}~
  	\Delta_{\ell}\ge \frac{\lambda_{0}^{2}k}{\ell}.
  	\end{align*} 
  \end{lemma}
  
  \begin{lemma}\label{lem:Emax_a}
  For every $q\in\mathbb{N}$ and $i\in[n]$, we have
  $	\bEE{\|a_{i}\|_{2}^{q}}
  	\le c_{0}^{q}(\Gamma(q))^{2}\lambda_{0}^{q}\bigg(\frac{k\sqrt{k\ell}}{m}\bigg)^{q}$.
          Further, when $\log k \geq 2$, it follows that $\bEE{\text{max}_{i\in[n]}\|a_{i}\|_{2}^{2}}\leq 
  	n^{\frac{2}{\log k}}\bEE{\|a_{1}\|_{2}^{\log k}}^{\frac{2}{\log k}}$. 
  \end{lemma}
  The proof of Lemma \ref{lem:ET_spectrum} is provided in Appendix \ref{app:ET_spectrum} and the proof of Lemma~\ref{lem:Emax_a} appears in Appendix~\ref{s:moment_bound}. 
  We close this section with the proof of Theorem \ref{thm_ecc}. 
  \begin{proof}[Proof of Theorem~\ref{thm_ecc}]
Recall that the proof is conditioned on exact recovery of the union ${\cS}_{\text{un}}$. Further, for notational simplicity, we set $\cS_{\text{un}}=[k\ell]$.
 We divide the proof into two steps.
 
  \noindent {\textit{Step 1.  Relating  probability of error to perturbation.}}
    Denote the event that Algorithm~\ref{algo_mutliple_disj} labels more
        than $\varepsilon k\ell$ coordinates incorrectly by $\cE$. An upper bound on $\bPr{\cE}$ would imply an upper bound on the probability of the error event implied by \eqref{err_eps}.
        The per support errors across the $\ell$ labels can have significant overlap or even be equal, so the criterion in \eqref{err_eps} is a good indicator of the number of misclustered coordinates determined by $\cE$. Additionally, it satisfies the triangle inequality, a property we will use later in proving Lemma \ref{lem:median_trick}.

    The following result
     relates the error probability  to a perturbation bound.
  
  \begin{lemma}[Error to perturbation bound]\label{l:error_perturbation}
    Let $V$ and $\hat{V}$, respectively, be $k\ell\times \ell$ matrices with $i$th column
    given by ${v}_i$ and $\hat{v}_{i}$, $1\leq i \leq \ell$,
    where $v_1, \ldots, v_\ell$
and $\hat{v}_1, \ldots, \hat{v}_\ell$
    denote the normalized eigenvectors of
    $\bEE{T}$ and $T$, respectively,
    corresponding to their $\ell$ largest eigenvalues.
  Then, 
  \begin{align}\label{e:prob_err_2}
  \bPr{\cE}\le \bPr{\|\hat{V}-VO\|_{F}\ge \frac 1 2 \sqrt{\frac{\ep \ell}2}},
  \end{align}
  where $O\in\mR^{\ell\times \ell}$ is a random orthonormal matrix and the probability on the right hand side is over the joint distribution of $\hat{V}$ and $O$.
  \end{lemma}
The proof of this lemma builds on the analysis in~\cite{Rohe_annals_stats_2011} and requires
  us to use some properties of $V$, which we
  note in the lemma below. 
  \begin{lemma}[Properties of $V$]\label{lem:V_prop}
    For $1\leq i\leq k\ell$, denote by $v^i$ the $i$th row of $V$. Then,
    the following properties hold:
    \begin{enumerate}
   \item {\em (Identity of rows of $V$ capture the partition)} $v^i=v^j$ if and only if $i$ and $j$ belong to the same
     support, i.e., $i,j\in\cS_{t}$ for some $t\in[\ell]$.     
  
    \item {\em (Minimum distance property)} For any two distinct rows $v^i$ and $v^j$, $\|v^i - v^j\|_2^2\geq 2/k$. 
  \end{enumerate}    
  \end{lemma}  
  We provide the proof of Lemma~\ref{lem:V_prop} in Appendix~\ref{app:lem_V_prop}.
  
  \begin{proof}[Proof of Lemma~\ref{l:error_perturbation}]
    We begin by observing that it suffices to show that
    \begin{align}
      \bPr{\cE}\le \bPr{\|C-VO\|_{F}\ge \sqrt{\frac{\varepsilon\ell}{2}}},\label{e:mislabel_bound1}
    \end{align}
  where $C$ is the matrix found in Step 6 of  Algorithm~\ref{algo_mutliple_disj}
  and is random since $\hat{V}$ is random. Indeed, 
  by Lemma~\ref{lem:V_prop}, $V$ has $\ell$ distinct rows,
  whereby $VO$, too, has $\ell$ distinct rows since $O$ is orthonormal. That is, $VO\in \cU_\ell$.
  Therefore, by triangle inequality, we get
  \begin{align}
    \|C-VO\|_F&\leq \|C-\hat{V}\|_F+ \|VO-\hat{V}\|_F
  \\
  &= \min_{U\in \cU_\ell} \|U-\hat{V}\|_F+\|VO-\hat{V}\|_F
  \\
  &\le 2\|VO-\hat{V}\|_F,
  \end{align}
  where the final bound holds since $VO$ belongs to $\cU_\ell$. Thus,~\eqref{e:mislabel_bound1} will imply~\eqref{e:prob_err_2}. Note that even if the matrix $O$ were to depend on $V$ and $\hat{V}$ and therefore be random, the result above holds with probability one, and the only property we require from $O$ is orthonormality.
  
  It remains to establish~\eqref{e:mislabel_bound1}. To that end, we define
  \begin{align}
  \cI\ed\{i\in[k\ell]:\|v^{i}O-c^{i}\|_{2}<1/\sqrt{2k}\},
  \end{align}
  where $v^{i}$ and $c^{i}$ are the $i$th row of $V$ and $C$, respectively.
  Our claim is that Algorithm~\ref{algo_mutliple_disj} does not make an error in
  labeling the coordinates in $\cI$, unless $|\cI^c|>\ep k \ell$.
  To see this, note that for any two distinct indices $i,j\in\cI$ we have
  \begin{align}
    \|v^{i}O-v^jO\|_2&\leq  \|v^{i}O-c^j\|_2+\|v^{j}O-c^j\|_2
    \\
    &\leq \|v^{i}O-c^{i}\|_{2}+\|c^{i}-c^{j}\|_{2}+\|v^{j}O-c^{j}\|_{2}\\
    &< \sqrt{\frac 2k} + \|c^i-c^j\|_2.
  \end{align}
  Thus, if $c^i=c^j$, we must have $\|v^iO-v^jO\|_2<\sqrt{2/k}$, which by
  the second property in Lemma~\ref{lem:V_prop} implies that
  $v^iO=v^jO$. 
  Therefore, when the labels given by the algorithm for coordinates $i$
  and $j$ coincide (this happens only when $c^i=c^j$), 
  then $v^iO=v^jO$. But then, by the first property in 
  Lemma~\ref{lem:V_prop}, the coordinates $i$ and $j$ must have been
  in the same part of $\cS$.  

  We have shown that the indices in $\cI$ that are assigned the same label by
  the algorithm must come from the same part in $\cS$. We still need to verify
  that coordinates from the same part in $\cS$ do not get assigned to
  different parts. We show this cannot happen unless $|\cI^c|>\ep k
  \ell$, and this is where we use the assumption that
  $\ep<1/\ell$. Indeed, if {$|\cI^c|\leq \ep k \ell< k$}, then at least
  one element from each part $\cS_1, \ldots,\cS_\ell$ must be in $\cI$,
  since $|\cS_i|=k$ for every $i$. By our previous observation, elements in
  each of these parts in $\cI$ must be assigned different labels by the
  algorithm, which means that it must assign at least $\ell$ different
  labels to the elements in $\cI$. Thus, if the algorithm assigns two
  elements in the same part $\cS_i$ different labels, it will assign
  more that $\ell$ different labels, which is not allowed.
  
  Therefore,
  all the indices in $\cI$ are correctly labeled when $|\cI^c|\leq \ep
  k\ell$. Then, clearly, in this case the error event $\cE$ does not
  hold. It follows from the definition of $\cI$ that
  \begin{align}
    \bPr{\cE} &\leq  \bPr{|\cI^c|>\ep k \ell}
    \\
    &\leq \bPr{\bigg\lvert\left\{i: \|c^i-v^iO\|_2
      \geq\frac 1{\sqrt{2k}} \right\}\bigg\rvert>\ep
      k \ell}
    \\
    &\leq \bPr{\|C-VO\|_F^2>\frac {\ep\ell}{2}},
  \end{align}  
  where in the final step we used the fact that the second step implies $\|C-VO\|_F^2=\sum_{i=1}^{k\ell} \|c^i-v^iO\|_2^2\ge \varepsilon k\ell/2k$.
  This completes the proof of~\eqref{e:mislabel_bound1}. 
  
  \end{proof}  
  
  \noindent {\textit{Step 2: Controlling the perturbation.}}
  
  In view of Lemma~\ref{l:error_perturbation}, we only need to control
  the perturbation $\|\hat{V}-{V}O\|_{F}$. We do this using the
  following extension of the Davis-Kahan theorem, which also fixes
  the choice of $O$.
  	\begin{theorem}[Perturbation of eigenspace]\cite{Yu_Biometrika_2015}\label{thm_DK_multiple}
  		Let $A$ and $\hat{A}$ be $d\times d$ symmetric matrices with eigenvalues $\nu_{1}\ge\cdots\ge \nu_{d}$ and $\hat{\nu}_{1}\ge\cdots\ge \hat{\nu}_{d}$, respectively. Let $V$ and $\hat{V}$ be $d\times \ell$ matrices consisting of the $\ell$ leading normalized eigenvectors of $A$ and $\hat{A}$, respectively.
  		Then, there exists an orthonormal matrix $O\in\mR^{\ell\times \ell}$ such that
  		\begin{align}
  		\|\hat{V}-{V}O\|_{F}^{2}
  		\le 2\sqrt{2}~\frac{\min\{\sqrt{\ell}\|\hat{A}-A\|_{op},\|\hat{A}-A\|_{F}\}}{\nu_{\ell}-\nu_{\ell+1}}.
  		\end{align}
  		
  	\end{theorem}
  By applying this result with $T$ and $\bEE{T}$ in the role of $\hat{A}$ and
  $A$, respectively, we get that there exists an orthonormal matrix $O$ such that
  \begin{align}
  		\|\hat{V}-{V}O\|_{F}\le 
  	\frac{2\sqrt{2}}{\Delta_{\ell}}\min\{\sqrt{\ell}\|T-\bEE{T}\|_{op},\|T-\bEE{T}\|_{F}\},
  		\end{align}
  where $\Delta_{\ell}\ed\nu_{\ell}-\nu_{\ell+1}$. Combining this bound
  with our earlier bound from Lemma~\ref{l:error_perturbation}, we get
  \begin{align}
  \bPr{\cE} &\leq
  \bPr{\|T-\bEE{T}\|_{op}  \geq \frac{\Delta_{\ell}\sqrt{\ep}}8}
    \\
    &\leq \frac{8 }{\Delta_{\ell}\sqrt{\ep}} \cdot \bEE{\|T-\bEE{T}\|_{op}},
  \label{e:error_expected_perturbation}
  \end{align}
  where the last step uses Markov's inequality.
  
  To bound the expected value on the right hand side, we use the following
  extension of a result of Rudelson~\cite{Rudelson_JFA_1999}. As pointed out
  earlier, the original bound in~\cite{Rudelson_JFA_1999} was
  restricted to isotropic $Z_i$s, and we show that it extends to
  arbitrary i.i.d. $Z_i$s with an extra factor. 
  The proof is provided in Appendix \ref{app:Rudelson_pf}. 
  \begin{theorem}[Extension of a result in \cite{Rudelson_JFA_1999}]\label{thm_Rud_alternate}
  		Let $Z\in\mR^{N}$ be a random vector such that $A=\bEE{ZZ^{\top}}$. Let $Z_{1},\ldots,Z_{n}$
  		be independent copies of $Z$. Then, there exists an absolute constant $c$ such that
  		\begin{align}\label{Rud_alternate}
  		\bEE{\bigg\|\frac{1}{n}\sum_{i=1}^{n}Z_{i}Z_{i}^{\top}-A\bigg\|_{op}}
  		\le \frac 12\left(\alpha^2 + \alpha\sqrt{\alpha^2+4\|A\|_{op}}\right),
  		\end{align}
                  where
                  \[
  \alpha= c\sqrt{\frac{ \bEE{\max_{i\in[n]}\|Z_i\|_2^2}\log N}n}.
                  \]
  	\end{theorem}
  Using this bound in~\eqref{e:error_expected_perturbation} with $N=k\ell$, we
  obtain 
  \begin{align}
    \bPr{\cE}\leq \frac{4}{\Delta_{\ell}\sqrt{\varepsilon}}\left(\alpha^2 + \alpha\sqrt{\alpha^2+4\|\bEE{T}\|_{op}}\right).
  \end{align}
  The proof is completed upon noting that $\alpha$ can be made
  smaller than $1/2$ using $n\ge c \bEE{\max_{i\in[n]}\|a_i\|_2^2}\, \log k\ell$, in which case
  $\alpha\sqrt{\alpha^2+4\|\bEE{T}\|_{op}}\leq
  \alpha\sqrt{8\max\{1,\|\bEE{T}\|_{op}\}}$. The error probability above can thus be made less than $1/4$ if $n\ge c (\log k\ell) \max\{1,\|\bEE{T}\|_{op}\}
  \bEE{\max_{i\in[n]}\|a_i\|_2^2}/(\Delta_{\ell}^{2}\varepsilon)$.
  \end{proof}

  In the next section, we combine the results from Theorems~\ref{thm:supp_union} and  \ref{thm_ecc} to show the sample complexity bound of Theorem~\ref{thm_multiple_supp_disj}.
  \subsection{Proof of Theorem \ref{thm_multiple_supp_disj}} \label{proof_thm_multiple_supp_disj}
  The proof of Theorem \ref{thm_multiple_supp_disj} now follows by combining guarantees for the union recovery step from Theorem \ref{thm:supp_union} and the clustering step from Theorem \ref{thm_ecc}.
  
  We begin by applying Theorem~\ref{thm:supp_union} to get that $\hat{\cS}_{\text{un}}$ coincides
  with $\cS_{\text{un}}=\cup_{i=1}^\ell\cS_i$ with probability close to
  $1$. Throughout, we condition on this event occurring. However, to
  avoid technical difficulties, we assume that a different set of independent
  samples is used to recover $\cS_{\text{un}}$ than those used to
  recover $\cS_1, \ldots, \cS_\ell$ -- thus, the overall number of samples needed
  will be the sum of samples needed for union recovery in Theorem~\ref{thm:supp_union}
  and the sample complexity determined in our analysis below.
  In particular, the
  clustering step dominates the sample complexity of our algorithm.
  
  Next, upon substituting the bounds from Lemma~\ref{lem:ET_spectrum} and
  Lemma~\ref{lem:Emax_a} into Theorem \ref{thm_ecc}, we see that for $\varepsilon$-approximate recovery of the supports it suffices to have
  	\begin{align}
  	n&\ge
  	\frac{c}{\varepsilon}\lambda_{0}^{2} \frac{k^{3}\ell}{m^{2}}\frac{\ell^{2}}{\lambda_{0}^{4}k^{2}}\cdot	n^{\frac{2}{\log k}}\cdot \bigg(\lambda_{0}\frac{k\sqrt{k}\sqrt{\ell}}{m}(\log k)^{2}\bigg)^{2}\cdot\log(k\ell)\nonumber \\
  	&=\frac{c}{\varepsilon} \frac{k^{4}\ell^{4}}{m^{4}}n^{\frac{2}{\log k}}(\log k)^{4}\log(k\ell).
  	\end{align}
 For $n\ge c((1/\varepsilon)(k\ell/m)^{4}\cdot (\log k)^{4}\log(k\ell))$,
  $n^{\frac{1}{\log k}}=O(1)$, which completes the proof in view of the
  sufficient condition for $n$ above.

  \section{Simulations}\label{sec:simulation}
  \subsection{Synthetic data}
  
\begin{figure}
\centering
\begin{subfigure}{.5\textwidth}
  \centering
  \includegraphics[width=.8\linewidth]{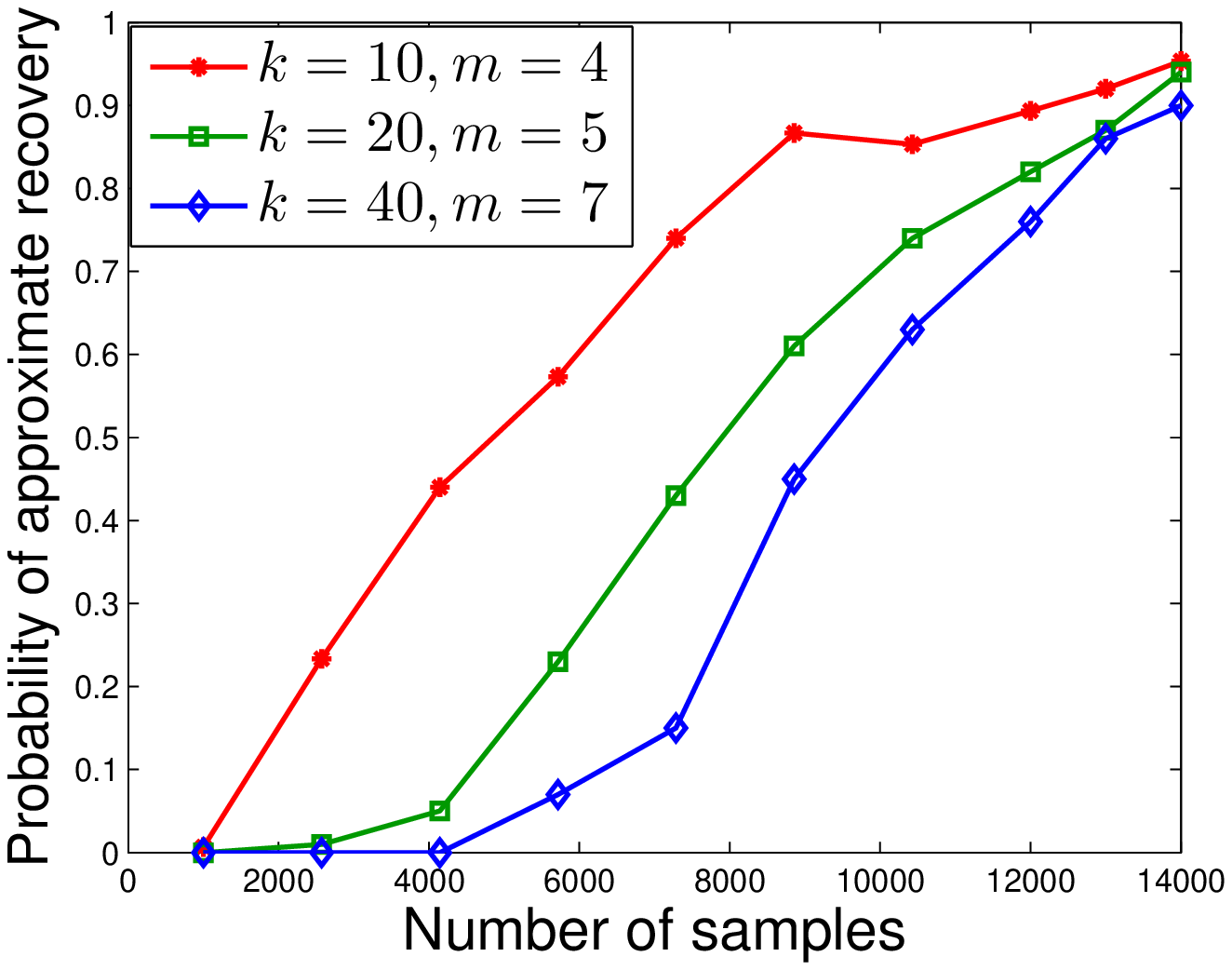}
  \caption{ $d=100$, $\varepsilon=0.2$, $\ell=2$.}
  \label{fig:sim_support}
\end{subfigure}%
\begin{subfigure}{.5\textwidth}
  \centering
  \includegraphics[width=.8\linewidth]{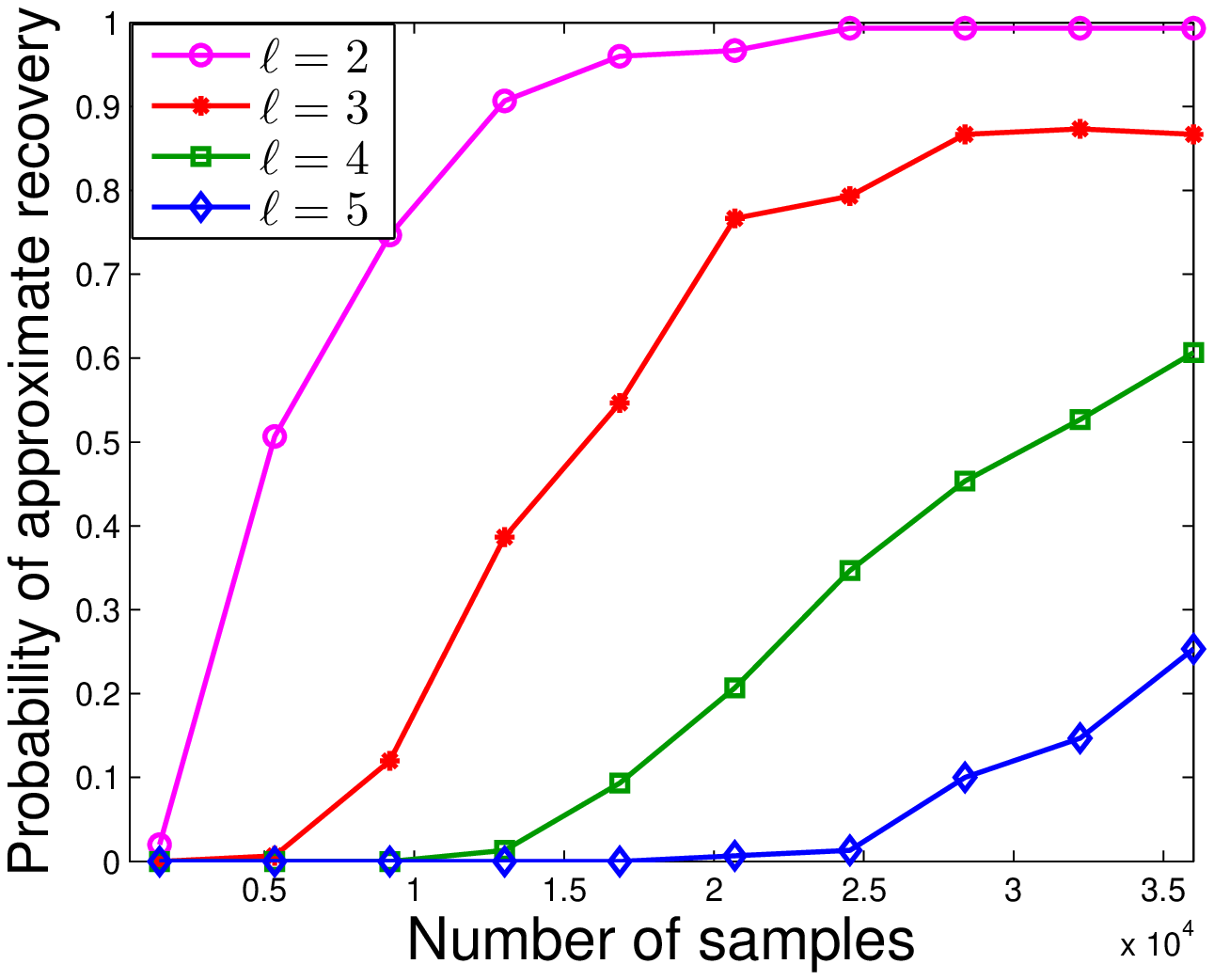}
  \caption{ $d=100$, $\varepsilon=0.2$, $m=4$, $k=10$.}
\label{fig:sim_support2}
\end{subfigure}
\caption{Probability of approximate support recovery with (a) varying $k/m$ ratios, and (b) varying $\ell$.}
\end{figure}
  
%
  
In this subsection, we evaluate the performance of Algorithm \ref{algo_mutliple_disj} on synthetic data for various parameter values.
Through these simulations, our goal is to see how the performance of the algorithm varies as a function of the ratio $k/m$ and $\ell$ for a fixed~$d$. 

 We first choose $d=100$, $\ell=2$ and consider three
  different values of $k/m$.  
  We generate two disjoint subsets $\cS_{1}$ and $\cS_{2}$ of
  $[d]$, each of size $k$. Then, for a given $n$, we generate $n/2$
  samples with each support, with values on the support drawn from the
  standard normal distribution in $\mR^{k}$. Measurement matrices
  $\{\Phi_{i}\}_{i=1}^{n}$ are generated independently with
  i.i.d. $\cN(0,1/m)$ entries and multiplied with the samples to obtain
  measurements $\{Y_{i}\}_{i=1}^{n}$. These measurements are given as input to
  the support recovery algorithm, which produces estimates for the union,
  as well as the individual supports, which we denote by $\hat{\cS}_{1}$ and
  $\hat{\cS}_{2}$.
  For each value of $(k,m,n)$, we run $100$ trials and declare it a
  success if the error $\sum_{i=1}^{2}|\hat{\cS}_{i}\Delta
  \cS_{\sigma(i)}|< 2\varepsilon k$. The plot in Figure \ref{fig:sim_support} shows the success rate
  over the $100$ trials as a function of the number of samples $n$, with $\varepsilon$ set as $0.2$. Note that the number of measurements taken per sample, $m$, is much
  smaller than the support size, $k$, of each sample. We can see from Figure \ref{fig:sim_support} that for a fixed probability of success, the number of samples required increases with $k/m$, which agrees with the result in Theorem~\ref{thm_multiple_supp_disj}. 
  In Figure \ref{fig:sim_support2}, we show the variation in the probability of approximate recovery as a function of $n$ for the number of supports $\ell = \{2, 3, 4, 5\}$, with $k$ and $m$ (and hence their ratio) held fixed. We can see that the number of samples required to achieve a given probability of recovery increases with $\ell$. Our current experiments however do not reveal whether the dependence on these parameters is tight.
%

  \subsection{MNIST dataset}
As an application involving natural data, we consider the problem of reconstructing handwritten images from very few linear measurements.
  We apply the multiple support recovery algorithm to the MNIST dataset \cite{mnist}, which consists of $60,000$ images of handwritten digits, each of size $28\times 28$.
  Each (grayscale) image is a sample in our setting, and the support of the sample essentially identifies the digit.
  This dataset fits well into our hypothesis that there is a small set of unknown supports underlying the data -- handwritten images corresponding to the same digit can be thought of as having roughly the same pattern (support) in the pixel domain. Thus, the vectorized version of images of the same digit will have approximately the same support.  
  We note that the task here is to recover the images of the digits from low dimensional projections, and not to learn a classifier using the dataset.

  In our experiments, the vectorized version of each image (a $784\times 1$ vector) is projected onto $m=100, 200$ or $500$ dimensions using Gaussian measurement matrices described in Assumption \ref{assump_phi}. Given these low dimensional projections, the goal is to identify the underlying digits. We fix $\ell=2$ and consider the example of digits $1$ and $5$ as shown in Figure~\ref{fig:digits_support}. The support size of each digit is roughly in the range $150-200$. It can be seen that Algorithm \ref{algo_mutliple_disj} can identify the distinct digits even when $m<k$. For comparison, we used the Group LASSO algorithm on the projected samples, which tries to recover the individual samples (images) itself. However, it requires a much larger number of measurements per sample (for example, about $m=500$ in this case).
  In fact, previously known algorithms for sparse recovery do not perform well in the low measurement regime of $m<k$, and we have used Group LASSO as an example to illustrate this fact. 
  
   We note that since these are handwritten digits, the support of samples coming from the {same} digit can also vary to some extent. However, the averaging across samples in our estimator takes care of this problem. Further, the supports from different digits need not be disjoint. To handle overlaps, we use the observation that $\tilde{\lambda}$ can provide an estimate for the intersection of supports as well. The plot of sorted entries of $\tilde{\lambda}$ shows a sharp drop in values at two locations, one  around the intersection and another around the union. We include this estimate of intersection of supports into our final estimate.
   This method performs well in practice, as can be seen in the results of Figure \ref{fig:digits_support}, where digits $1$ and $5$ have significant overlap.
   
  \begin{figure}[th]
  				\centering
  				\begin{subfigure}{.22\textwidth}
  				  \centering
  				  \includegraphics[height=2cm,width=3cm]{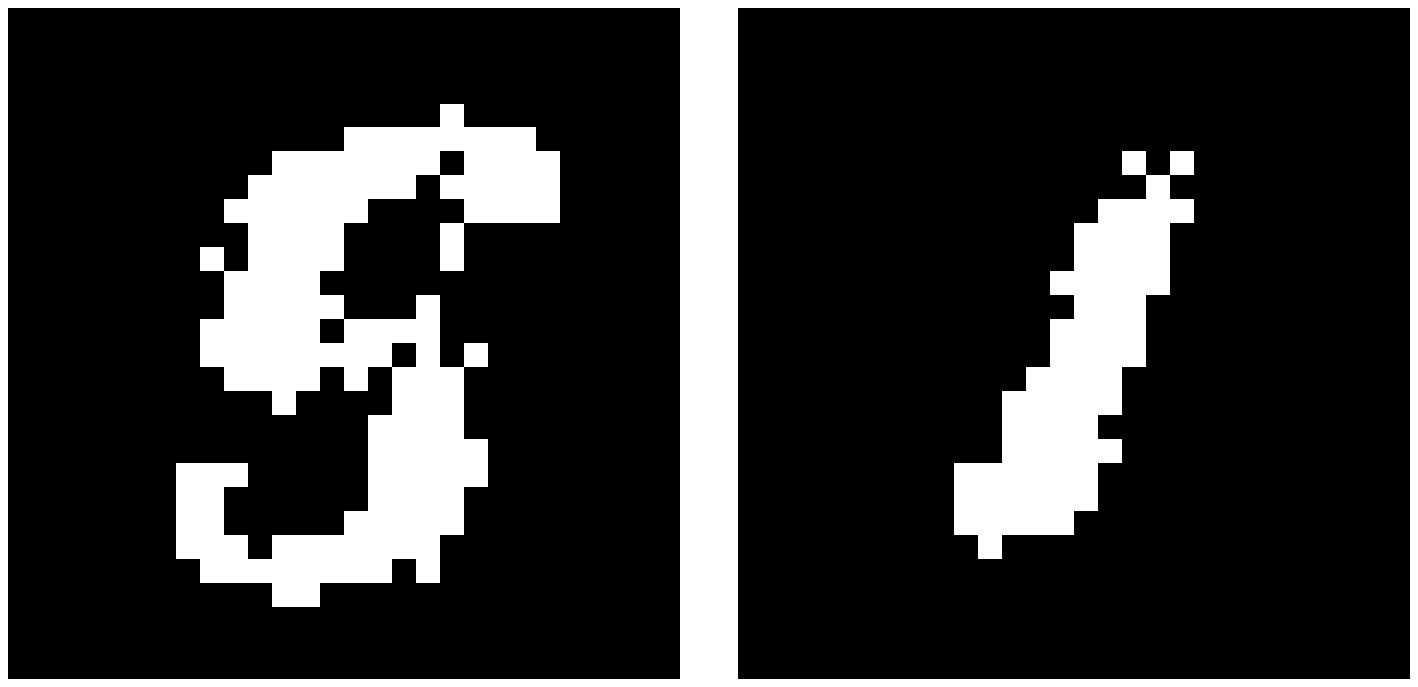}
  				 \caption{$m=100,n=2000$}
  				\end{subfigure}
  				\quad
  				\begin{subfigure}{.22\textwidth}
  				  \centering
  				  \includegraphics[height=2cm,width=3cm]{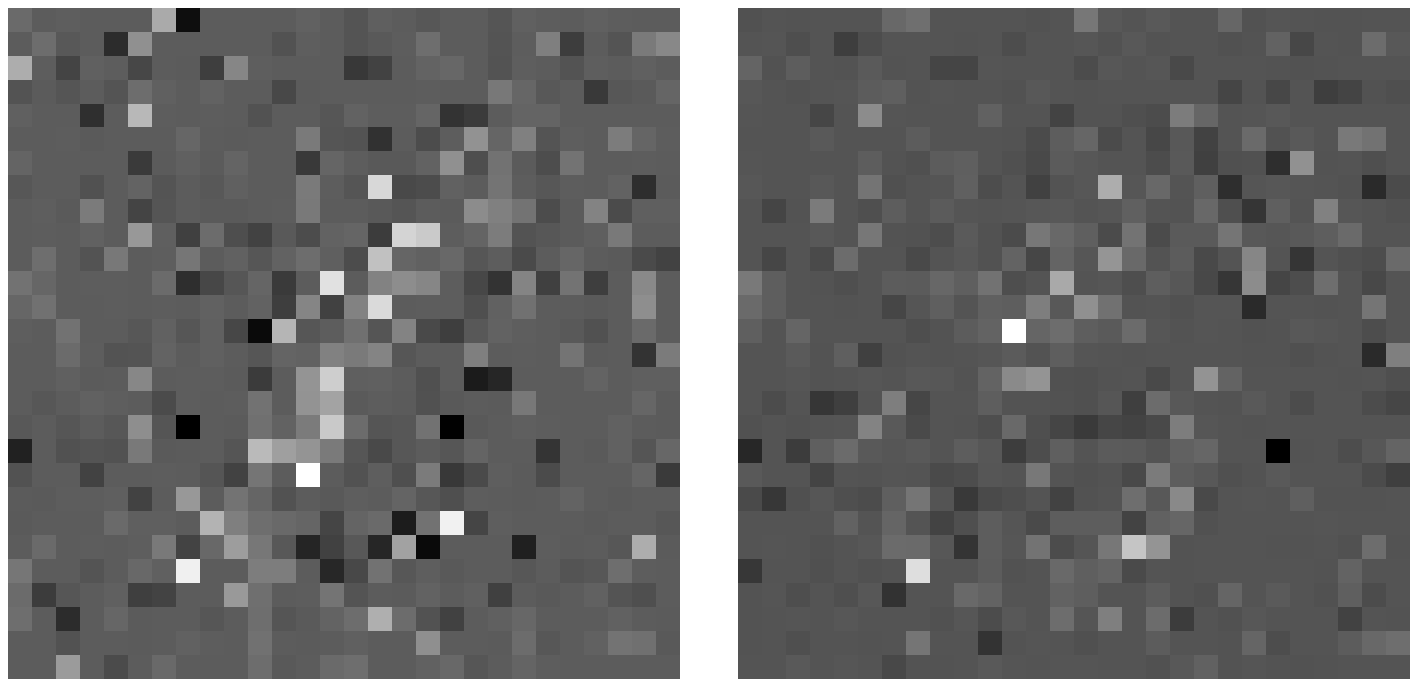}
  				  \caption{$m=100,n=2000$}
  				\end{subfigure}
  				
  				\begin{subfigure}{.22\textwidth}
  				  \centering
  				 	\includegraphics[height=2cm,width=3cm]{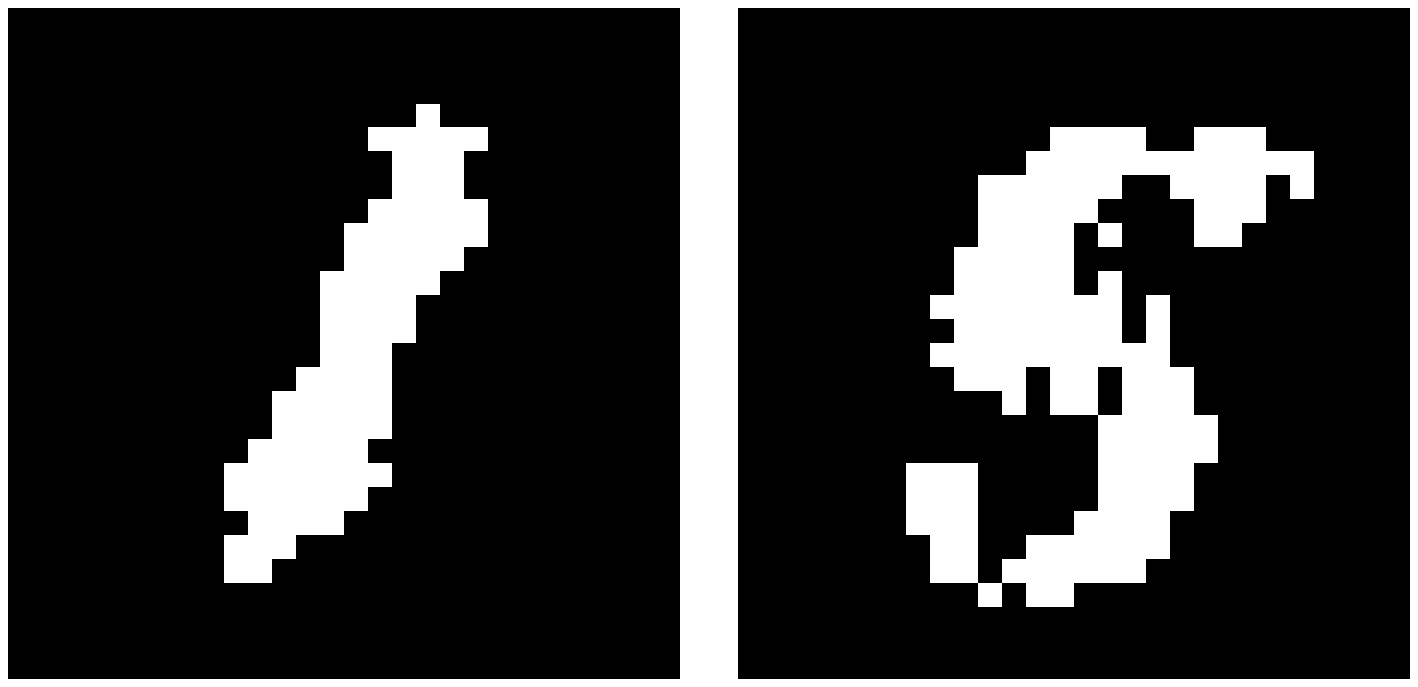}
  				  \caption{$m=200,n=2000$}
  				\end{subfigure}
  				\quad 
  					\begin{subfigure}{.22\textwidth}
  						  \centering
  						 	\includegraphics[height=2cm,width=3cm]{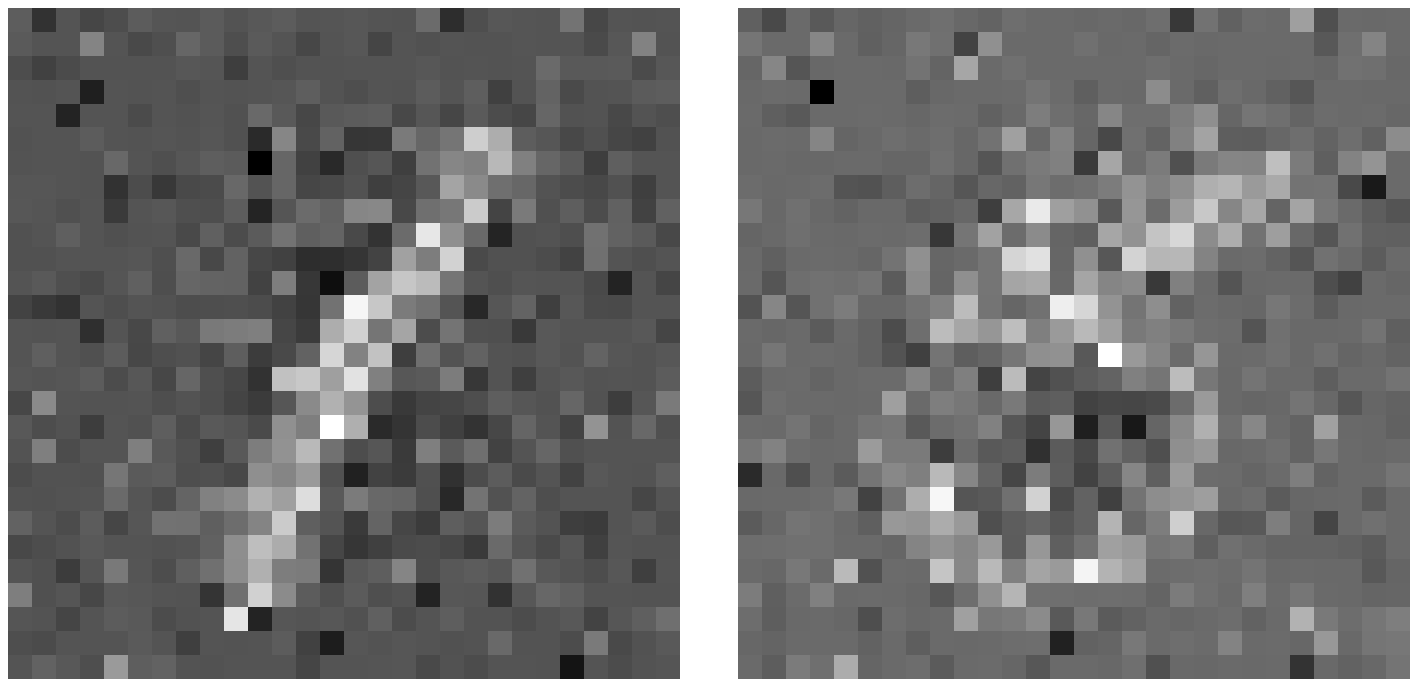}
  						  \caption{$m=200,n=2000$}
  						\end{subfigure}
  						
  					\centering
  					\begin{subfigure}{.22\textwidth}
  					  \centering
  					  \includegraphics[height=2cm,width=3cm]{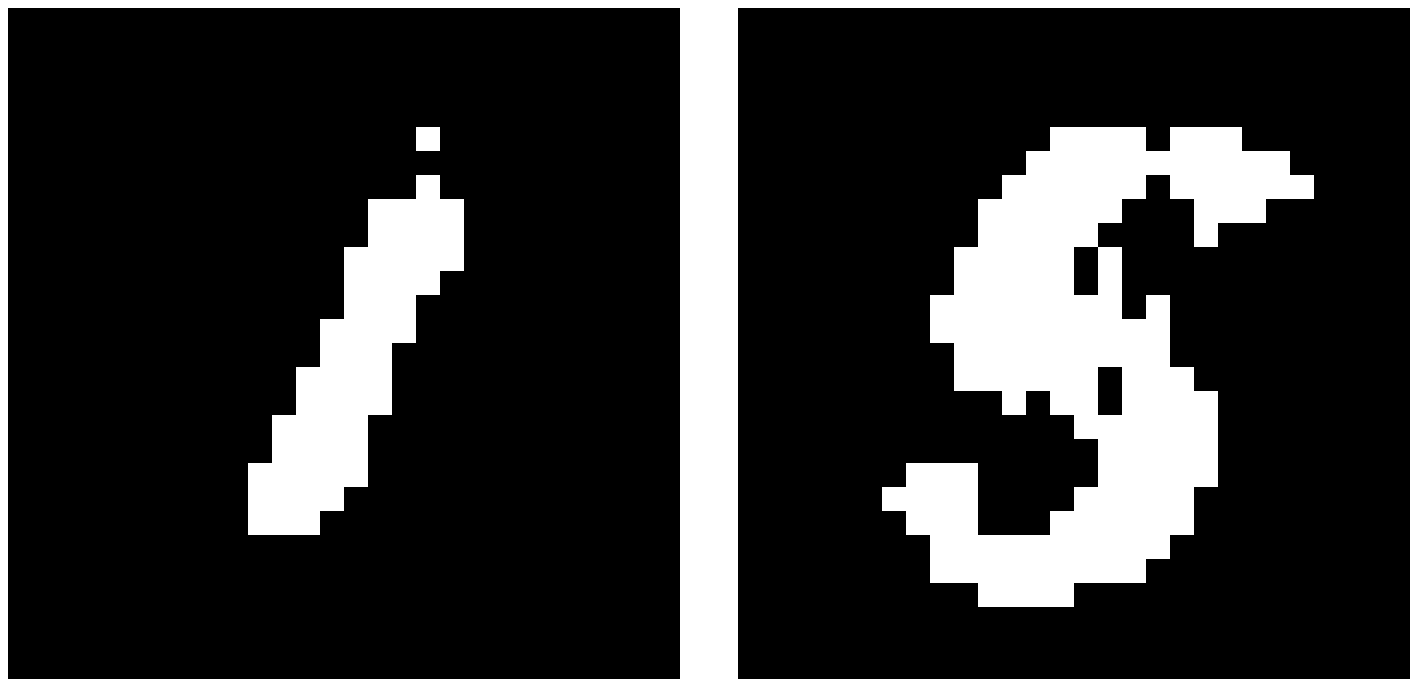}
  					 \caption{$m=500,n=2000$}
  					\end{subfigure}
  					\quad
  					\begin{subfigure}{.22\textwidth}
  					  \centering
  					  \includegraphics[height=2cm,width=3cm]{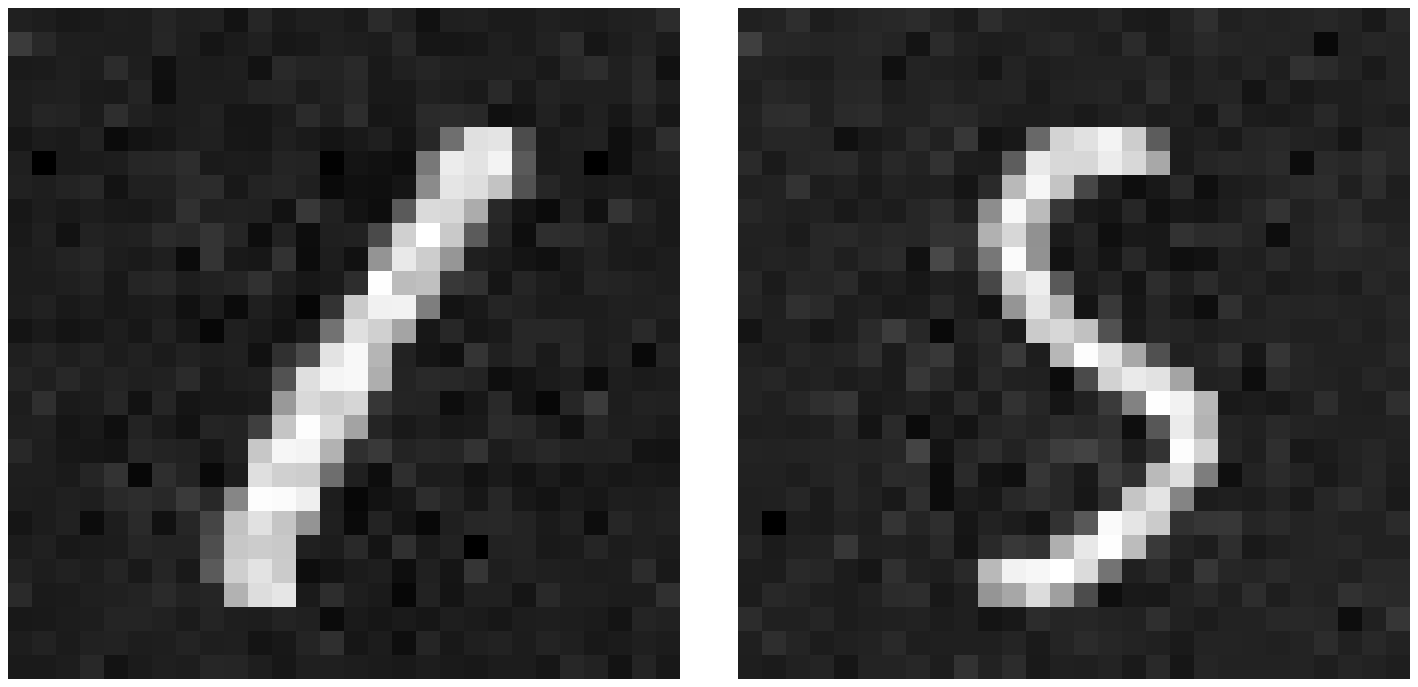}
  					  \caption{$m=500,n=2000$}
  					\end{subfigure}
  							  \caption{Recovery performance of Algorithm \ref{algo_mutliple_disj} ((a),(c),(e)), and Group LASSO ((b),(d),(f)). 
  							  } \label{fig:digits_support} 
  							  \end{figure}
 \vspace{-1cm} 							  
\subsection{Computational complexity}

The first step in our algorithm for estimating the union involves computing the average variance along each of the $d$ coordinates and requires $O(mnd)$ operations.
The clustering step involves computing the $T$ matrix and its $\ell$ leading eigenvectors which requires $O(k^{3}\ell^{3}+k^{2}\ell^{2}n)$ operations, followed by the $\ell$-means step which requires $O(k\ell^{3})$ operations per iteration.  Other algorithms for recovering multiple supports do not perform well when $m<k$, and have computational complexity that scales quadratically or worse with $d$. For instance, the sparse Bayesian learning based algorithm from \cite{Wang_UAI_2015} has a complexity of $O(d^{2})$ per iteration, and LASSO-based procedures have a complexity of $O(d^{2})$ or $O(d^{3})$ per iteration, depending on the specific algorithm used.								  
  							  
  \section{Discussion}

Throughout in this work, we assumed that the distinct supports were pairwise disjoint sets.
In the case of overlapping supports, the structure of the expected affinity matrix, and consequently its spectrum, changes. For the special case of $\ell=2$, overlapping supports can be handled by a simple modification of the sign-based estimate. Instead of partitioning the coordinates in the union estimate based on the sign of the eigenvector, we now use a threshold $\tau>0$ and declare coordinates with values in $[-\tau,\tau]$ as belonging to both supports  (values above $\tau$ or below $-\tau$ are assigned to different supports). The optimal $\tau$ can be explicitly characterized in terms of the parameters of the problem.
Given our current algorithm, a simple way to handle this case for {\em{general}} $\ell$ would be to use fuzzy $\ell$-means, which returns scores for each coordinate indicating how likely it is to belong to a certain support. However, choosing a threshold to decide the supports using the scores is difficult in general. 
Some other approaches have been explored in the graph clustering literature, but these do not apply directly to our setting. Other extensions of this work include studying the performance of the algorithm under different support sizes, and prior distribution with non-uniform mixing weights. Also, our work shows a sufficient condition on the number of samples required for multiple support recovery; obtaining the necessary condition is a challenging task in general and requires characterizing the distance between mixture distributions. Using a component wise distance bound leads to the same lower bound as in \cite{Ramesh_arxiv_2019} (with an additional $1/\ell$ factor), and obtaining
a better lower bound seems difficult.

  

\begin{appendices}

\section{Remaining proofs from Section \ref{sec:analysis_cluster}}\label{app:remaining_lemmas_cluster}

\subsection{Proof of Lemma~\ref{lem:median_trick}~(Probability of error boosting)}\label{app:median_trick}

Given an $(n,\varepsilon,1/4)$-estimator for $\Sigma_{k,\ell,d}$, we apply it to $L$ independent blocks of data. Specifically,
denoting this estimator by $e$,
	consider independent copies $(Y^n(t), \Phi^n(t))$, $1\leq t \leq L$, of $(Y^n,\Phi^n)$. For $t\in [L]$, let
	\[
	(\hat{\cS}_{1, t} , \ldots, \hat{\cS}_{\ell,t}):=e(Y^n(t), \Phi^{n}(t))
	\]
	denote the output for the estimator applied to the $t$th block.

        We now describe a procedure to output a final estimate for the supports using the estimates $(\hat{\cS}_{1,t},\ldots,\hat{\cS}_{\ell,t})$ from the $L$ blocks of samples. For each $t\in[L]$, we check if there
is a set $\cI\subseteq[L]\backslash\{t\}$ of cardinality $N\ge L/2$ satisfying
\begin{align}\label{e:condition2}
\min_{\sigma_{t}\in \cG_{\ell}}~\frac{1}{k\ell}\sum_{i=1}^{\ell}|\hat{\cS}_{i,t}\Delta\hat{\cS}_{\sigma_{t}(i),t^\prime}|\le 2\varepsilon, \quad \forall\,t^{\prime}\in\cI.
\end{align}
That is, we look for a $t$ for which $(\hat{\cS}_{1,t},\ldots,\hat{\cS}_{\ell,t})$ are close to $L/2$ other estimates. This indicates ``robustness'' of the estimate from the $t$th block, making it an appropriate proxy for the median.
Our final estimate is $(\bar{\cS}_{1},\ldots,\bar{\cS}_{\ell})=(\hat{\cS}_{1,t},\ldots,\hat{\cS}_{\ell,t})$, where $t$ is an index which satisfies the property above.

We show that for $L\geq \lceil 8\ln \frac 1 \delta\rceil$ the estimator
above constitutes an $(nL, 3\varepsilon, \delta)$-estimator for $\Sigma_{k,\ell,d}$. 
Indeed, denoting
	\[
	Z_t=\mathds{1}\left(\exists\, \sigma\in \cG_\ell \text{ s.t. }
	\frac 1{k\ell}\sum_{i=1}^\ell|\cS_i \Delta \hat{\cS}_{\sigma(i),t}|\leq \ep
	\right),
	\]
	by our assumption for the estimator $e$ we have
	\[
	\bE{\bPP{(\cS_1, \ldots, \cS_\ell)}}{Z_t}\geq \frac 3 4.
	\]
	Furthermore, $Z_t$ are independent for different $t\in [L]$. Thus, by Hoeffding's inequality,
	\[
	\bP{(\cS_1, \ldots, \cS_\ell)}{\sum_{t=1}^L Z_t\leq \frac{L}{2}}\leq e^{-\frac{L}{8}}, \quad
	\forall\, (\cS_1, \ldots, \cS_\ell)\in \Sigma_{k,\ell,d}.
	\]
	In particular, for $L\geq \lceil 8\ln \frac 1 \delta\rceil$,
	with probability exceeding $1-\delta$
	there exist\footnote{Without loss of generality, we assume $L$ to be even.} $M\geq L/2+1$
	indices $t_1, \ldots, t_{M}\in [L]$ and permutations $\sigma_1, \ldots, \sigma_M\in \cG_\ell$ such that
	\begin{align}
	\frac 1{k\ell}\sum_{i=1}^\ell |\cS_i \Delta \hat{\cS}_{\sigma_j(i),t_j}|\leq \ep, \quad \forall\,j\in[M].
	\label{e:condition1}
	\end{align}
        Note that since $|A\Delta B|$ is a metric for subsets of $[d]$, the estimate $(\hat{\cS}_{1,t},\ldots,\hat{\cS}_{\ell,t})$ for $t=t_1$ satisfies~\eqref{e:condition2} when~\eqref{e:condition1} holds; in fact, any index among $\{t_1, ..., t_M\}$ can serve this purpose. 
        However, the estimate described earlier need not select any of these indices. Yet, we now show that any other index chosen by the procedure will work as well, provided~\eqref{e:condition1} holds.

        To that end, denote by $\cI^\prime$ the set $\{t_1, \ldots, t_M\}$ of indices satisfying~\eqref{e:condition1}, and recall the set $\cI$ found by our estimation procedure earlier. Then, when $|\cI^\prime|\ge L/2+1$, which holds with probability exceeding $1-\delta$,
        \[
	|\cI \cap \cI^\prime|\geq |\cI|+|\cI^\prime| - L \geq 1,
	\]
whereby 
	there exists an index $t\in [L]$ and permutations $\sigma, \overline{\sigma}  \in \cG_\ell$ such that
	\begin{align*}
	\frac 1{k\ell}\sum_{i=1}^\ell |\cS_i \Delta \hat{\cS}_{\sigma(i),t}|\leq \ep\,\, \text{ and }\,\,
	\frac 1{k\ell}\sum_{i=1}^\ell |\overline{\cS}_i \Delta \hat{\cS}_{\overline{\sigma}(i),t}|\leq 2\ep.
	\end{align*}
        It follows that the permutation $\sigma^\prime=\sigma \circ \overline{\sigma}^{-1}$ satisfies
	\[
	\frac 1{k\ell}\sum_{i=1}^\ell |\cS_i \Delta \overline{\cS}_{\sigma^\prime(i)}|\leq 3\ep,
	\]
	which completes the proof.
	\qed

\subsection{Proof of Lemma~\ref{lem:Emax_a}}\label{s:moment_bound}
As noted in the proof of Theorem \ref{thm_multiple_supp_disj}, 
the clustering step in our algorithm is analyzed under the assumption that the union of supports is exactly recovered in the first step, whereby we can set
 $\hat{\cS}_{\text{un}}=\cS_{\text{un}}$. 
 
 We will first show the bound on $\bEE{\max_{i\in[n]}\|a_{i}\|_{2}^{2}}$, followed by the moment bound for $\bEE{\|a_{i}\|_{2}^{q}}$. We start by noting that for any $q\ge2$,
 	\begin{align*}
 	\bEE{\max_{i\in[n]}\|a_{i}\|_{2}^{2}}
 	&=\bEE{\bigg(\max_{i\in[n]}\|a_{i}\|_{2}^{q}\bigg)^{\frac{2}{q}}}\\
 	&\le \bEE{\bigg(\sum_{i=1}^{n}\|a_{i}\|_{2}^{q}\bigg)^{\frac{2}{q}}}\\
 	&\le \bigg(\bEE{\sum_{i=1}^{n}\|a_{i}\|_{2}^{q}}\bigg)^{\frac{2}{q}}\\
 	&=n^{\frac{2}{q}}\bigg(\bEE{\|a_{1}\|_{2}^{q}}\bigg)^{\frac{2}{q}},
 	\end{align*}
 	where we used Jensen's inequality in the third step.
    For $\log
         k\ge 2$, upon setting $q=\log k$ in the inequality above, we get
 	\begin{align*}
 	\bEE{\max_{i\in[n]}\|a_{i}\|_{2}^{2}}\leq 
 	n^{\frac{2}{\log k}}\bigg(\bEE{\|a_{1}\|_{2}^{\log k}}\bigg)^{\frac{2}{\log k}}.
 	\end{align*}
 
 We now proceed to bound $\bEE{\|a_{i}\|_{2}^{q}}$. In the rest of the proof, we will denote $a_{i}\in\mR^{d}$ by $a$, and with some abuse of notation, denote by $\Phi_{i}$ the $i$th column of $\Phi$.
 By using the definition of $a$, we have	
		\begin{align*}
		\|a\|_{2}^{2q}
		=\bigg(\sum_{i\in{\cS}_{\text{un}}}a_{i}^{2}\bigg)^{q}
		&=\bigg(\sum_{i\in{\cS}_{\text{un}}}(\Phi_{i}^{\top}\Phi_{\cS} X_{\cS})^{4}\bigg)^{q}\\
		&=\bigg(\sum_{i\in{\cS}_{\text{un}}}(\alpha_{i}^{\top} X_{\cS})^{4}\bigg)^{q}\\
		&=\bigg(\sum_{i\in{\cS}_{\text{un}}}(X_{\cS}^{\top}A_{i}X_{\cS})^{2}\bigg)^{q},
		\end{align*}
		where $\alpha_{i}=\Phi_{\cS}^{\top}\Phi_{i}$ as defined before and $A_{i}\ed \alpha_{i}\alpha_{i}^{\top}$.
		To compute the expectation of the term in the last step, we first condition on $\Phi$ and note that
			\begin{align}\label{chaos_sum_moment}
				\mathbb{E}\bigg[\bigg(\sum_{i\in{\cS}_{\text{un}}}(X_{\cS}^{\top}A_{i}X_{\cS})^{2}\bigg)^{q}\bigg\rvert\Phi\bigg]
				&=(k\ell)^{q}\mathbb{E}\bigg[\bigg(\frac{1}{k\ell}\sum_{i\in{\cS}_{\text{un}}}(X_{\cS}^{\top}A_{i}X_{\cS})^{2}\bigg)^{q}\bigg\rvert\Phi\bigg]\nonumber\\
				&\le (k\ell)^{q-1}\sum_{i\in{\cS}_{\text{un}}}\bEE{(X_{\cS}^{\top}A_{i}X_{\cS})^{2q}|\Phi},
			\end{align}
where we used $|{\cS}_{\text{un}}|=k\ell$, and the convexity of the function $x^{q}$ for $x\ge 0$, $q\in\mN$.
The quantity on the right essentially involves the $(2q)$th moment of a subexponential random variable (see Appendix \ref{app:moment_concentration} for definition).  
 To see that the
quadratic form $X_{\cS}^{\top}A_{i}X_{\cS}$
is subexponential, we use
the Hanson-Wright inequality ($cf.$~\cite{Rudelson_ECP_2013})
to get
		\begin{align*}
		\mathbb{P}&(|X_{\cS}^{\top}A_{i}X_{\cS}-\mu|\ge t|\Phi)
		\le 2\exp\bigg(-\min\bigg\{\frac{t^{2}}{\lambda_{0}^{2}\|A_{i}\|_{F}^{2}},\frac{t}{\lambda_{0}\|A_{i}\|_{op}}\bigg\}\bigg), 
		\end{align*}
		where $\mu=\bEE{X_{\cS}^{\top}A_{i}X_{\cS}|\Phi}=\lambda_{0}\|\alpha_{i}\|_{2}^{2}$. Lemma \ref{lem_sbx_moment} in  Appendix \ref{app:moment_concentration} can now be used to bound the moment in \eqref{chaos_sum_moment}.
Specifically, we get
		\begin{align*}
		\mathbb{E}[(X_{\cS}^{\top}A_{i}X_{\cS})^{2q}|\Phi]
		&\le 2q\cdot(16)^{q}\bigg(\Gamma(q)\lambda_{0}^{2q}\|A_{i}\|_{F}^{2q}
		+\Gamma(2q)\lambda_{0}^{2q}\|A_{i}\|_{op}^{2q}\bigg)
		+2^{2q}\mu^{2q}\\
		&\le 3q\cdot(16)^{q}\Gamma(2q)\lambda_{0}^{2q}\|\alpha_{i}\|_{2}^{4q},
		\end{align*}
		where we used $\|A_{i}\|_{F}=\|A_{i}\|_{op}=\|\alpha_{i}\|_{2}^{2}$.
              Next, taking expectation over $\Phi$, we obtain
		\begin{align}\label{chaos_moment}
		\bEE{(X_{\cS}^{\top}A_{i}X_{\cS})^{2q}}
		&\le c^{\prime}_{q}\Gamma(2q)\lambda_{0}^{2q}\bEE{\|\alpha_{i}\|_{2}^{4q}},
		\end{align}
		where $c^{\prime}_{q}=3q\cdot (16)^{q}$.
		Thus, combining the result above with \eqref{chaos_sum_moment}, we get
		\begin{align}
		\mathbb{E}\bigg[\bigg(\sum_{i\in{\cS}_{\text{un}}}(X_{\cS}^{\top}A_{i}X_{\cS})^{2}\bigg)^{q}\bigg]
		&\le c^{\prime}_{q}\Gamma(2q)\lambda_{0}^{2q}(k\ell)^{q}
		\sum_{i\in{\cS}_{\text{un}}}\bEE{\|\alpha_{i}\|_{2}^{4q}}
              \nonumber  \\
		&= c^{\prime}_{q}\Gamma(2q)\lambda_{0}^{2q}(k\ell)^{q}
		\bigg(
		\sum_{i\in \cS}\bEE{\|\alpha_{i}\|_{2}^{4q}}
		+\sum_{i\in {\cS}_{\text{un}}\backslash \cS}\bEE{\|\alpha_{i}\|_{2}^{4q}}\bigg).\label{chaos_sum_moment_split}
		\end{align}
          When $i\in \cS$,
		\begin{align*}
		\mathbb{E}[\|\alpha_{i}\|_{2}^{4q}] 
		&=\bEE{\bigg(\|\Phi_{i}\|_{2}^{4}+\sum_{j\in \cS\backslash\{i\}}(\Phi_{i}^{\top}\Phi_{j})^{2}\bigg)^{2q}}\\
		&\le 2^{2q}\left(\bEE{\|\Phi_{i}\|_{2}^{8q}}
		+\bEE{\bigg(\sum_{j\in \cS\backslash\{i\}}(\Phi_{i}^{\top}\Phi_{j})^{2}\bigg)^{2q}}\right),
		\end{align*}
		and when $i\in {\cS}_{\text{un}}\backslash \cS$,
		\begin{align*}
		\bEE{\|\alpha_{i}\|_{2}^{4q}}
		&\le 
		\bEE{\bigg(\sum_{j\in \cS}(\Phi_{i}^{\top}\Phi_{j})^{2}\bigg)^{2q}}.
		\end{align*}
		Since $\Phi_{i}$ has independent, subgaussian entries with parameter $1/m$, we see that $\|\Phi_{i}\|_{2}^{2}\sim\sbx(c^{\prime}/m,c^{\prime\prime}/m)$ with $c^{\prime}=128$ and $c^{\prime\prime}=8$ \cite[Lemma D.2]{Ramesh_arxiv_2019}. This gives, using Lemma \ref{lem_sbx_moment},
		\begin{align*}
		\bEE{(\|\Phi_{i}\|_{2}^{2})^{4q}}
		&\le 2q(16)^{q}\bigg(\Gamma(2q)\frac{c^{\prime 2q}}{m^{2q}}+\Gamma(4q)\frac{c^{\prime\prime 4q}}{m^{4q}}\bigg)
		+(\bEE{\|\Phi_{i}\|_{2}^{2}})^{4q}\nonumber\\
		&\le 4q(16)^{q}c^{\prime 2q}\Gamma(4q)\frac{1}{m^{2q}}+1,
		\end{align*}
		where we used $c^{\prime}>c^{\prime\prime 2}$.
		Using similar arguments, we note that $\Phi_{i}^{\top}\Phi_{j}|\Phi_{i}$ is subgaussian with parameter $\|\Phi_{i}\|_{2}^{2}/m$, which implies that,
conditioned on $\Phi_{i},$                 
                $\sum_{j\in \cS\backslash\{i\}}(\Phi_{i}^{\top}\Phi_{j})^{2}$
is
$\sbx(c^{\prime}(k-1)\|\Phi_{i}\|_{2}^{4}/m^{2},c^{\prime\prime}\|\Phi_{i}\|_{2}^{2}/m)$.
Then, using Lemma \ref{lem_sbx_moment} again, we get
		\begin{align*}
		\mathbb{E}\bigg[\bigg(\sum_{j\in \cS\backslash\{i\}}(\Phi_{i}^{\top}\Phi_{j})^{2}\bigg)^{2q}\bigg]
		&\le~c^{\prime}_{q}\bE{\Phi_{i}}{\Gamma(q)c^{\prime q}\bigg(\frac{k-1}{m^{2}}\bigg)^{q}\|\Phi_{i}\|_{2}^{4q}+\Gamma(2q)c^{\prime\prime 2q}\bigg(\frac{\|\Phi_{i}\|_{2}^{2}}{m}\bigg)^{2q}}\nonumber \\
		&~+2^{2q}\bigg(\bEE{\sum_{j\in \cS\backslash\{i\}}(\Phi_{i}^{\top}\Phi_{j})^{2}}\bigg)^{2q}\\
		&\le~ c^{\prime}_{q}c^{\prime q}\Gamma(q)\bigg(\frac{k-1}{m^{2}}\bigg)^{q}
		\bigg(1+2c^{\prime}_{q}c^{\prime 2q}\Gamma(2q)\frac{1}{m^{q}}\bigg)\nonumber\\
		&+c^{\prime}_{q}c^{\prime\prime 2q}\Gamma(2q)\frac{1}{m^{2q}}\bigg(1+c^{\prime}_{q}c^{\prime 2q}\Gamma(2q)\frac{1}{m^{q}}\bigg)+2^{2q}\bigg(\frac{k-1}{m}\bigg)^{2q}\\
		&\le~5c^{\prime}_{q}c^{\prime 2q}\Gamma(2q)\bigg(\frac{k}{m}\bigg)^{2q}.
		\end{align*}       
		Combining these results and substituting into \eqref{chaos_sum_moment_split}, we get 
		\begin{align*}
		\mathbb{E}\bigg[\bigg(\sum_{i\in{\cS}_{\text{un}}}(X_{\cS}^{\top}A_{i}X_{\cS})^{2}\bigg)^{q}\bigg]
		&\le c^{\prime}_{q}\Gamma(2q)\lambda_{0}^{2q}(k\ell)^{q-1}
		\bigg(\sum_{i\in{S}}\bEE{\|\alpha_{i}\|_{2}^{4q}}+\sum_{i\in\cS_{\text{un}}\backslash \cS}\bEE{\|\alpha_{i}\|_{2}^{4q}}\bigg)\\
		&\le 5c^{\prime 2}_{q}c^{\prime 2q}\Gamma(2q)\lambda_{0}^{2q}(k\ell)^{q-1}
		\bigg(k\Gamma(2q)\bigg(\frac{k}{m}\bigg)^{2q}
		+(k\ell-k)\Gamma(2q)\bigg(\frac{k}{m}\bigg)^{2q}\bigg)\\
		&= 5c^{\prime 2}_{q}c^{\prime 2q}(\Gamma(2q))^{2}\lambda_{0}^{2q} \bigg(\frac{k\sqrt{k\ell}}{m}\bigg)^{2q}.
		\end{align*}
		Rescaling the exponent, we get
		\begin{align*}
		\bEE{\|a\|_{2}^{q}}
		&=\bEE{\bigg(\sum_{i\in{\cS}_{\text{un}}}(X_{\cS}^{\top}A_{i}X_{\cS})^{2}\bigg)^{\frac{q}{2}}}\\
		&\le 5c_{q/2}^{2}c^{\prime q}(\Gamma(q))^{2}\lambda_{0}^{q}\bigg(\frac{k\sqrt{k\ell}}{m}\bigg)^{q}
		\end{align*}
Noting that $c^{\prime}(5c_{q/2}^{2})^{1/q}\le 45\cdot8c^{\prime}=c_{0}$, we obtain the result. 
              \qed

\subsection{Proof of Lemma \ref{lem:V_prop}}\label{app:lem_V_prop}
\begin{enumerate}[(i)]
\item To show the first property, we note that the true covariance matrix can be decomposed as $\bEE{T}=WBW^{\top}+(\dg-\on)I$, where $W\in\{0,1\}^{k\ell\times \ell}$ encodes the block structure, and $B\in\mR^{\ell\times \ell}$ contains the distinct values from each block. In particular,
  for $1\leq i \leq k\ell$ and $1\leq j \leq \ell$, define 
	\begin{align*}
	W_{ij}=\begin{cases}
	1,~\text{if}~i\in \cS_{j},\\
	0,~\text{otherwise}, 
	\end{cases}
	\end{align*}
	and, for $1\leq i \leq \ell$ and $1\leq j \leq \ell$, define 
	\begin{align*}
	B_{ij}=\begin{cases}
	\on,~\text{if}~i=j,\\
	\off,~\text{otherwise}.
	\end{cases}
	\end{align*}
Since $\bEE{T}$ and $WBW^{\top}$ have the same set of eigenvectors, we will show that the matrix $V\in\mR^{k\ell\times \ell}$ consisting of the $\ell$ leading eigenvectors of $WBW^{\top}$ has the desired property.
To that end, first note that there are only $\ell$ unique rows in $W$, one unique row corresponding to each block. We will show that  $V$ also consists of $\ell$ unique rows, in exact correspondence with the rows of $W$.
To do so, we will follow \cite[Lemma 3.1]{Rohe_annals_stats_2011} and show that $V$ is essentially a row-transformed version of $W$, i.e., there exists an invertible matrix $H\in\mR^{\ell\times \ell}$ such that $WH=V$. We start by considering the eigen decomposition
\begin{align*}
(W^{\top}W)^{\frac{1}{2}}B(W^{\top}W)^{\frac{1}{2}}
=U\Lambda U,
\end{align*}
where $\Lambda\in\mR^{\ell\times\ell}$ is diagonal and $U\in\mR^{\ell\times\ell}$ is an orthonormal matrix.
Left multiplying by $W(W^{\top}W)^{-\frac{1}{2}}$ and right multiplying by $(W^{\top}W)^{-\frac{1}{2}}W^{\top}$ in the equation above, we get,
\begin{align*}
WBW^{\top}=WH\Lambda (WH)^{\top},
\end{align*}
where $H\ed (W^{\top}W)^{-\frac{1}{2}}U$.
Finally, right multiplying by $WH$ and noting that $(WH)^{\top}WH=I$, we have
\begin{align*}
WBW^{\top}\cdot WH=WH\cdot \Lambda,
\end{align*}
implying that the columns of $WH$ are the normalized eigenvectors of $WBW^{\top}$. 

We have thus shown that $V=WH$. Let $v^{i}$ and $w^{i}$ denote the $i$th row of $V$ and $W$, respectively.
If $v^{i}=v^{j}$ for some $i\ne j$, then $w^{i}H=w^{j}H$. Since $H=(W^{\top}W)^{-\frac{1}{2}}U$ is invertible, this implies $w^{i}=w^{j}$. Conversely, if $w^{i}=w^{j}$ for some $i\ne j$, then $w^{i}H=w^{j}H$, which implies $v^{i}=v^{j}$.	

\item 	
	Using the fact that $V=WH$ from (i), we have for $v^{i}\ne v^{j}$,
		\begin{align*}
		\|v^{i}-v^{j}\|_{2}
		&=\|(w^{i}-w^{j})H\|_{2}\\
		&\ge \sqrt{2}\nu_{\min}(H),
		\end{align*}
		where $\nu_{\min}(H)\ed\min_{\|x\|_{2}=1}\|x^{\top}H\|_{2}$, and we used $\|w^{i}-w^{j}\|_{2}=\sqrt{2}$ for $w^{i}\ne w^{j}$. Now,
		\begin{align*}
		\min_{\|x\|_{2}=1}\|x^{\top}H\|_{2}^{2}
		&=\min_{\|x\|_{2}=1}x^{\top}HH^{\top}x\\
		&=\min_{\|x\|_{2}=1}x^{\top}(WW^{\top})^{-1}x\\
		&=\frac{1}{k},
		\end{align*}
		where we used $HH^{\top}=(W^{\top}W)^{-\frac{1}{2}}UU^{\top}(WW^{\top})^{-\frac{1}{2}}=(WW^{\top})^{-1}$ and the fact that $WW^{\top}=k~\diag{I}$. Putting everything together, we get
		\begin{align*}
		\|v^{i}-v^{j}\|_{2}^{2}\ge \frac{2}{k}.
		\end{align*}

\end{enumerate}	

\subsection{Proof of Theorem \ref{thm_Rud_alternate}}\label{app:Rudelson_pf}

The proof is similar to that of \cite{Rudelson_JFA_1999}, and we highlight the steps needed to extend the result to general $A$.	In particular, following similar arguments as in \cite{Rudelson_JFA_1999}, it can be shown that
		\begin{align}
				\mathbb{E}\bigg[\bigg\|\frac{1}{n}\sum_{i=1}^{n}Z_{i}Z_{i}^{\top}-A\bigg\|_{op}\bigg]
				\le c\frac{\sqrt{\log N}}{n}\sqrt{\bEE{\max_{i\in[n]}\|Z_{i}\|_{2}^{2}}}\sqrt{\bEE{\bigg\|\sum_{i=1}^{n}Z_{i}Z_{i}^{\top}\bigg\|_{op}}},\label{eq_op}
				\end{align}
		Now,
		\begin{align}
		\bEE{\bigg\|\sum_{i=1}^{n}Z_{i}Z_{i}^{\top}\bigg\|_{op}}
		&\le n\bEE{\bigg\|\frac{1}{n}\sum_{i=1}^{n}Z_{i}Z_{i}^{\top}-A\bigg\|_{op}+\|A\|_{op}}\nonumber \\
		&= n(\beta+\|A\|_{op}),\label{eq_beta}
		\end{align}
		where $\beta\ed \bEE{\bigg\|\frac{1}{n}\sum_{i=1}^{n}Z_{i}Z_{i}^{\top}-A\bigg\|_{op}}$.
		It follows from \eqref{eq_op} and \eqref{eq_beta} that
		\begin{align*}
		\beta\le c\sqrt{\frac{\log N}{n}}\sqrt{\bEE{\max_{i\in[n]}\|Z_{i}\|_{2}^{2}}}\sqrt{\beta +\|A\|_{op}}.
		\end{align*}
		Letting $\alpha=c\sqrt{(\log N)/n}\sqrt{\bEE{\max_{i\in[n]}\|Z_{i}\|_{2}^{2}}}$, we have the solution
		\begin{align*}
		\beta\le \frac{1}{2}\bigg(\alpha^{2}+\alpha\sqrt{\alpha^{2}+4\|A\|_{op}}\bigg),
		\end{align*}
which completes the proof.


\subsection{Proof of Lemma \ref{lem:ET_structure}}\label{app:ET_structure}
Our goal is to compute the expected value of the  clustering matrix, denoted $\bEE{T}$, and we will do so by first conditioning on the measurement ensemble $\bphi_{1}^{n}$ and noting that each entry of $T$ is then of the form $(X^{\top}AX)^{2}$, where $X$ is subgaussian and $A$ is a fixed matrix (given $\bphi_{1}^{n}$). This conditional expectation can be calculated using Lemma \ref{lem_qform_moment}. The next step is to average over the distribution of $\bphi_{1}^{n}$, and our analysis will require the moment assumptions on the entries of $\bphi_{1}^{n}$ described in Assumption~$2$. 
		Although each entry of $\bEE{T}$
		can be explicitly characterized in terms of the system parameters, we will sometimes only mention the leading terms. In fact, the analysis of our algorithm in Theorem \ref{thm_multiple_supp_disj} only requires an upper bound on the diagonal entries and tight upper and lower bounds on the off diagonal entries of $\bEE{T}$.
		
Specifically, by the definition of $T$ from \eqref{Tuv}, we note that
		\begin{equation}\label{Tuv}
		\bEE{T_{uv}}=\frac{1}{n}\sum_{j=1}^{n}(\bphi_{ju}^{\top}\bphi_{j}X_{j})^{2}\cdot (\bphi_{jv}^{\top}\bphi_{j}X_{j})^{2},
		\end{equation} 
		for $(u,v)\in{\cS}_{\text{un}}\times {\cS}_{\text{un}}$.
		The expectation in the expression above is over the joint distribution of $X_{1}^{n}$, $\bphi_{1}^{n}$ and the labels $G_{1}^{n}$ (generating samples from the mixture  $\dP_{\cS}=\frac 1 \ell \sum_{i=1}^{\ell} \dP^{(i)}$ described in Section II in the main file can be thought of as drawing the label $G$ uniformly from $[\ell]$, and conditioned on $G=g$, drawing a sample from $\dP^{(g)}$). We will first condition on the labels (or, equivalently, on the random subsets $\{I_{1},\ldots,I_{\ell}\}$ defined as $I_{i}\ed \{j\in[n]: \mathtt{supp}(X_{j})=\cS_{i}\}$ and on the measurement matrices.
		We focus on a single summand in \eqref{Tuv}, and drop the dependence on the sample index $j$. With a slight abuse of notation, we let $\cS=\text{supp}(X)$ denote the support of the sample we focus on and note that
		\begin{align*}
		\bE{X}{(\bphi_{u}^{\top}\bphi X)^{2}\cdot (\bphi_{v}^{\top}\bphi X)^{2}|\bphi,G}
		=\bE{X}{(X_{\cS}^{\top}\alpha_{u}\alpha_{v}^{\top}X_{\cS})^{2}|\bphi,G},
		\end{align*}
		where,  $\alpha_{u}\ed \bphi_{\cS}^{\top}\bphi_{u}$, $u\in{\cS}_{\text{un}}$. We can now use Lemma \ref{lem_qform_moment} to get
		\begin{align}\label{M_offdiag_phi}
		\bE{X}{(X_{\cS}^{\top}\alpha_{u}\alpha_{v}^{\top}X_{\cS})^{2}|\bphi,G}
		=\rho\sum_{i\in{\cS}}\alpha_{ui}^{2}\alpha_{vi}^{2}
		+\lambda_{0}^{2}\sum_{i\ne j}\alpha_{ui}^{2}\alpha_{vj}^{2} +\lambda_{0}^{2}\sum_{i\ne j}\alpha_{ui}\alpha_{vi}\alpha_{uj}\alpha_{vj},
		\end{align}
		where recall $\lambda_{0}=\bEE{X_{i}^{2}}$ and $\rho=\bEE{X_{i}^{4}}$. We will first handle the $u=v$ case, which will be used to compute the diagonal entries of the mean matrix. We have, for every $u\in{\cS}_{\text{un}}$,
		\begin{align*}
		\bE{X,\bphi}{(X_{\cS}^{\top}\alpha_{u}\alpha_{u}^{\top}X_{\cS})^{2}|G}
		&=\rho\bE{\bphi}{\sum_{i\in{\cS}}\alpha_{ui}^{4}|G}+2\lambda_{0}^{2}\bE{\bphi}{\sum_{i\ne j}\alpha_{ui}^{2}\alpha_{uj}^{2}|G}\\
		&=
		\rho\bE{\bphi}{\sum_{i\in{\cS}}(\bphi_{u}^{\top}\bphi_{i})^{4}|G}
		+2\lambda_{0}^{2}\bE{\bphi}{\sum_{i\ne j}(\bphi_{u}^{\top}\bphi_{i})^{2}(\bphi_{u}^{\top}\bphi_{j})^{2}|G}.
		\end{align*}

		When $u\in \cS$,
		\begin{align}
		  \dg^{s}&\ed\bE{X,\bphi}{(X_{\cS}^{\top}\alpha_{u}\alpha_{u}^{\top}X_{\cS})^{2}|G}
                 \nonumber \\
		&=\rho\bE{\bphi}{\|\bphi_{u}\|_{2}^{8}+\sum_{i\in{\cS}\backslash\{u\}}(\bphi_{u}^{\top}\bphi_{i})^{4}|G}
		+2\lambda_{0}^{2}\bE{\bphi}{2\|\bphi_{u}\|_{2}^{4}\sum_{i\in{\cS}\backslash\{u\}}(\bphi_{u}^{\top}\bphi_{i})^{2}+\sum_{i\ne j}(\bphi_{u}^{\top}\bphi_{i})^{2}(\bphi_{u}^{\top}\bphi_{j})^{2}|G}\nonumber\\
		&\le c\rho\bigg(1+\frac{k-1}{m^{2}}\bigg)
		+c^{\prime}\lambda_{0}^{2}\bigg(\frac{k-1}{m}+\frac{(k-1)(k-2)}{m^{2}}\bigg),\label{M_diag_big}
		\end{align}
		where we used Lemma \ref{lem_qform_moment} in the second step and Lemma \ref{lem_moments_ip} in the third step, and  retained the leading terms.

		When $u\in {\cS}_{\text{un}}\backslash \cS$, using Lemmas \ref{lem_qform_moment} and \ref{lem_moments_ip} once again, we have 
		\begin{align}
		  \dg^{d} &\ed\bE{X,\bphi}{(X_{\cS}^{\top}\alpha_{u}\alpha_{u}^{\top}X_{\cS})^{2}|G}
              \nonumber    \\
		&=\rho\bE{\bphi}{\sum_{i\in{\cS}}(\bphi_{u}^{\top}\bphi_{i})^{4}|G}+2\lambda_{0}^{2}\bE{\bphi}{\sum_{i\ne j}(\bphi_{u}^{\top}\bphi_{i})^{2}(\bphi_{u}^{\top}\bphi_{j})^{2}|G}\nonumber\\
		&\le c\rho\bigg(\frac{k}{m^{2}}\bigg)+c^{\prime}\lambda_{0}^{2}\frac{k(k-1)}{m^{2}}.\label{M_diag_small}
		\end{align}
		We now use these results to bound the diagonal entries of the mean matrix $\bEE{T}$. Using \eqref{Tuv}, \eqref{M_diag_big} and \eqref{M_diag_small}, we see that for $u\in \cS_{1}$,
		\begin{align}
		  \dg\ed \bEE{T_{uu}}               
		&=\bE{G}{\bE{X,\bphi}{\frac{1}{n}\bigg(\sum_{j\in I_{1}}(\bphi_{ju}^{\top}\bphi_{j}X_{j})^{4}
				+\cdots+\sum_{j\in I_{\ell}}(\bphi_{ju}^{\top}\bphi_{j}X_{j})^{4}\bigg)\bigg\rvert G}}\nonumber\\
		&=\bE{G}{\frac{1}{n}\bigg(|I_{1}|\dg^{s}+\sum_{i=2}^{\ell}|I_{i}|\dg^{d}\bigg)}\nonumber\\
		&=\frac{1}{\ell}\dg^{s}+\frac{\ell-1}{\ell}\dg^{d}\nonumber\\
		&\le  \frac{c}{\ell}\bigg\{\rho\bigg(1+\frac{k-1}{m^{2}}\bigg)
		 +\lambda_{0}^{2}\bigg(\frac{k-1}{m}+\frac{(k-1)(k-2)}{m^{2}}\bigg)\bigg\}
		+ \frac{c(\ell-1)}{\ell}\bigg\{\rho\bigg(\frac{k}{m^{2}}\bigg)+\lambda_{0}^{2}\frac{k(k-1)}{m^{2}}\bigg\},\label{M_diag}
		\end{align}
		where we used $\bE{G}{|I_{i}|}=n/\ell$ for all $i\in[\ell]$, under the uniform mixture assumption. The same result holds for $u\in \cS_{i}$ for any $i\in[\ell]$.
		
		The next step is to bound the off diagonal entries of $\bEE{T}$. Continuing from \eqref{M_offdiag_phi}, we will handle each of the three terms separately. For each of these terms, we will consider the case when both $u$ and $v$ belong to the same support, and when they belong to different supports. Overall, these calculations highlight the block structure of $\bEE{T}$, with the diagonal entries all being equal, and the off diagonal entries taking two different values based on whether the indices belong to the same support or not.

     For the first term in \eqref{M_offdiag_phi}, when $(u,v)\in \cS\times \cS$, $u\ne v$, we have
		\begin{align}
		\bE{\bphi}{\sum_{i\in \cS}\alpha_{ui}^{2}\alpha_{vi}^{2}|G}
		=&\bE{\bphi}{\|\bphi_{u}\|_{2}^{4}(\bphi_{u}^{\top}\bphi_{v})^{2}|G}
		+\bE{\bphi}{\|\bphi_{v}\|_{2}^{4}(\bphi_{u}^{\top}\bphi_{v})^{2}|G}+\bE{\bphi}{\sum_{i\in \cS\backslash\{u\}\cup\{v\}}(\bphi_{i}^{\top}\bphi_{u})^{2}
			(\bphi_{i}^{\top}\bphi_{v})^{2}\bigg\rvert G}\nonumber\\
		=&\frac{2}{m}\bigg(1+\frac{3}{\nm}(c_2-1)+\frac{1}{\nm^{2}}(c_3-3c_2+2)\bigg)+\frac{k-2}{m^{2}}\bigg(1+\frac{1}{m}(c_{2}-1)\bigg)\ed\gamma_{1}^{s},\label{M_offdiag_1_big}
		\end{align}
		using Lemma \ref{lem_moments_ip}. On the other hand, when $(u,v)\in \cS\times {\cS}_{\text{un}}\backslash \cS$, we have
		\begin{align}
		\bE{\bphi}{\sum_{i\in \cS}\alpha_{ui}^{2}\alpha_{vi}^{2}|G}
		=&\bE{\bphi}{\|\bphi_{u}\|_{2}^{4}(\bphi_{u}^{\top}\bphi_{v})^{2}|G}
		+\bE{\bphi}{\sum_{i\in \cS\backslash\{u\}}(\bphi_{i}^{\top}\bphi_{u})^{2}
			(\bphi_{i}^{\top}\bphi_{v})^{2}\bigg\rvert G}\nonumber\\
		=&\frac{1}{m}\bigg(1+\frac{3}{\nm}(c_2-1)+\frac{1}{\nm^{2}}(c_3-3c_2+2)\bigg)+\frac{k-1}{m^{2}}\bigg(1+\frac{1}{m}(c_{2}-1)\bigg)\ed\gamma_{1}^{sd}.\label{M_offdiag_1_med}
		\end{align}
		The same result holds when $(u,v)\in{\cS}_{\text{un}}\backslash \cS\times \cS$. Finally, when $(u,v)\in {\cS}_{\text{un}}\backslash \cS \times {\cS}_{\text{un}}\backslash \cS$,
		\begin{align}
		\bE{\bphi}{\sum_{i\in \cS}\alpha_{ui}^{2}\alpha_{vi}^{2}|G}
		&=\bE{\bphi}{\sum_{i\in \cS}(\bphi_{i}^{\top}\bphi_{u})^{2}
			(\bphi_{i}^{\top}\bphi_{v})^{2}|G}\nonumber\\
		&=\frac{k}{m^{2}}\bigg(1+\frac{1}{m}(c_{2}-1)\bigg)\ed\gamma_{1}^{d}.\label{M_offdiag_1_small}
		\end{align}
		
		For the second term in \eqref{M_offdiag_phi}, when $(u,v)\in \cS\times \cS$, 
		\begin{align}
		  \bE{\bphi}{\sum_{i\ne j}\alpha_{ui}^{2}\alpha_{vj}^{2}|G}               
		&=\bE{\bphi}{\|\bphi_{u}\|_{2}^{4}\|\bphi_{v}\|_{2}^{4}
			+(\bphi_{u}^{\top}\bphi_{v})^{4}
			+\|\bphi_{u}\|_{2}^{4}\sum_{i\in \cS\backslash\{u\}\cup\{v\}}(\bphi_{v}^{\top}\bphi_{i})^{2}\bigg\rvert G}\nonumber \\
		&\quad +\bE{\bphi}{\|\bphi_{v}\|_{2}^{4}\sum_{i\in \cS\backslash\{u\}\cup\{v\}}(\bphi_{u}^{\top}\bphi_{i})^{2}
			+(\bphi_{u}^{\top}\bphi_{v})^{2}\sum_{i\in \cS\backslash\{u\}\cup\{v\}}(\bphi_{v}^{\top}\bphi_{i})^{2}\bigg\rvert G}\nonumber \\ 
		&\quad +\bE{\bphi}{(\bphi_{u}^{\top}\bphi_{v})^{2}\sum_{i\in \cS\backslash\{u\}\cup\{v\}}(\bphi_{u}^{\top}\bphi_{i})^{2}
			+\sum_{\substack{i,j\in \cS\backslash\{u\}\cup\{v\} \\ i\neq j}}(\bphi_{u}^{\top}\bphi_{i})^{2}\cdot (\bphi_{v}^{\top}\bphi_{j})^{2}\bigg\rvert G}\nonumber \\
		&=\bigg(1+\frac{1}{m}(c_{2}-1)\bigg)^{2}
		+\bigg(\frac{2}{m^{2}}+\frac{1}{m^{3}}(c_{2}^{2}-2)\bigg)
		+2\bigg(1+\frac{1}{m}(c_{2}-1)\bigg)\frac{k-2}{m}\nonumber \\
		&\quad +2\frac{(k-2)}{m^{2}}\bigg(1+\frac{1}{m}(c_{2}-1)\bigg)
		+\frac{(k-2)(k-3)}{m^{2}}\ed\gamma_{2}^{s},\label{M_offdiag_2_big}
		\end{align}
		where we used Lemma \ref{lem_moments_ip} in the second step. When $(u,v)\in \cS\times {\cS}_{\text{un}}\backslash \cS$, 
		\begin{align}
		\bE{\bphi}{\sum_{i\ne j}\alpha_{ui}^{2}\alpha_{vj}^{2}|G}
		&=\bE{\bphi}{\|\bphi_{u}\|_{2}^{4}\sum_{i\in \cS\backslash\{u\}}(\bphi_{v}^{\top}\bphi_{i})^{2}|G} 
		+\bE{\bphi}{(\bphi_{u}^{\top}\bphi_{v})^{2}\sum_{i\in \cS\backslash\{u\}}
			(\bphi_{u}^{\top}\bphi_{i})^{2}|G}\nonumber \\
		&\quad+\bE{\bphi}{\sum_{\substack{i,j\in \cS\backslash\{u\} \\ j\neq i}}(\bphi_{u}^{\top}\bphi_{i})^{2}\cdot (\bphi_{v}^{\top}\bphi_{j})^{2}|G} \nonumber\\
		&=\bigg(1+\frac{1}{m}(c_{2}-1)\bigg)\frac{k-1}{m}
		+\frac{(k-1)}{m^{2}}\bigg(1+\frac{1}{m}(c_{2}-1)\bigg) +\frac{(k-1)(k-2)}{m^{2}}\ed\gamma_{2}^{sd},\label{M_offdiag_2_med}
		\end{align}
		and the same expression holds when $(u,v)\in{\cS}_{\text{un}}\backslash \cS\times \cS$.
		When $(u,v)\in{\cS}_{\text{un}}\backslash \cS\times {\cS}_{\text{un}}\backslash \cS$, 
		\begin{align}
		\bE{\bphi}{\sum_{i\ne j}\alpha_{ui}^{2}\alpha_{vj}^{2}|G}
		&=\bE{\bphi}{\sum_{\substack{i,j\in \cS \\ j\neq i}}(\bphi_{u}^{\top}\bphi_{i})^{2}\cdot (\bphi_{v}^{\top}\bphi_{j})^{2}|G} =\frac{k(k-1)}{m^{2}}\ed\gamma_{2}^{d},\label{M_offdiag_2_small}
		\end{align}
		
		Finally, for the third term in \eqref{M_offdiag_phi}, when $(u,v)\in \cS\times \cS$,
		\begin{align}
		\bE{\bphi}{\sum_{i\ne j}\alpha_{ui}\alpha_{vi}\alpha_{uj}\alpha_{vj}|G}
		=&\bE{\bphi}{\|\bphi_{u}\|_{2}^{2}\bphi_{u}^{\top}\bphi_{v}\cdot \|\bphi_{v}\|_{2}^{2}\bphi_{u}^{\top}\bphi_{v}|G}\nonumber \\
		&+\bE{\bphi}{\|\bphi_{u}\|_{2}^{2}\bphi_{u}^{\top}\bphi_{v}\sum_{j\in \cS\backslash\{u\}\cup\{v\}}(\bphi_{u}^{\top}\bphi_{j})\cdot (\bphi_{v}^{\top}\bphi_{j})|G}\nonumber \\
		&+\bE{\bphi}{\|\bphi_{v}\|_{2}^{2}\bphi_{u}^{\top}\bphi_{v}\sum_{j\in \cS\backslash\{u\}\cup\{v\}}(\bphi_{u}^{\top}\bphi_{j})\cdot (\bphi_{v}^{\top}\bphi_{j})|G}\nonumber \\
		&+\bE{\bphi}{\sum_{\substack{i,j\in \cS\backslash\{u\}\cup\{v\} \\ j\neq i }}(\bphi_{u}^{\top}\bphi_{i}) (\bphi_{v}^{\top}\bphi_{i})(\bphi_{u}^{\top}\bphi_{j}) (\bphi_{v}^{\top}\bphi_{j})|G}\nonumber\\
		=&\frac{1}{m}\bigg(1+\frac{c_{2}-1}{m}\bigg)^{2}+\frac{2(k-2)}{m^{2}}\bigg(1+\frac{c_{2}-1}{m}\bigg)+\frac{(k-2)(k-3)}{m^{3}}\ed\gamma_{3}^{s}.\label{M_offdiag_3_big}
		\end{align}
		When $(u,v)\in \cS\times {\cS}_{\text{un}}\backslash \cS$,
		\begin{align}
		\bE{\bphi}{\sum_{i\ne j}\alpha_{ui}\alpha_{vi}\alpha_{uj}\alpha_{vj}|G}
		=&\bE{\bphi}{\|\bphi_{u}\|_{2}^{2}\bphi_{u}^{\top}\bphi_{v}\sum_{j\in \cS\backslash\{u\}}(\bphi_{u}^{\top}\bphi_{j})\cdot (\bphi_{v}^{\top}\bphi_{j})|G}\nonumber \\
		&+\bE{\bphi}{\sum_{\substack{i,j\in \cS\backslash\{u\} \\ j\neq i}}(\bphi_{u}^{\top}\bphi_{i}) (\bphi_{u}^{\top}\bphi_{j})(\bphi_{v}^{\top}\bphi_{i}) (\bphi_{v}^{\top}\bphi_{j})\bigg\rvert G}\nonumber\\
		=&\frac{(k-1)}{m^{2}}\bigg(1+\frac{c_{2}-1}{m}\bigg)
		+\frac{(k-1)(k-2)}{m^{3}}\ed\gamma_{3}^{sd},\label{M_offdiag_3_med}
		\end{align}
		and the same expression holds when $(u,v)\in {\cS}_{\text{un}}\backslash \cS \times \cS$. 
		When $(u,v)\in {\cS}_{\text{un}}\backslash \cS\times {\cS}_{\text{un}}\backslash \cS$,
		\begin{align}
		\bE{\bphi}{\sum_{i\ne j}\alpha_{ui}\alpha_{vi}\alpha_{uj}\alpha_{vj}|G}
		&=\bE{\bphi}{\sum_{\substack{i,j\in \cS \\ j\neq i}}(\bphi_{u}^{\top}\bphi_{i}) (\bphi_{u}^{\top}\bphi_{j})(\bphi_{v}^{\top}\bphi_{i}) (\bphi_{v}^{\top}\bphi_{j})\bigg\rvert G}\nonumber\\
		&=\frac{k(k-1)}{m^{3}}\ed\gamma_{3}^{d},\label{M_offdiag_3_small}
		\end{align}
		We have thus computed the expected values of each of
                the three terms in \eqref{M_offdiag_phi}.
                
		Thus, combining \eqref{M_offdiag_1_big}, \eqref{M_offdiag_2_big} and \eqref{M_offdiag_3_big} and using \eqref{M_offdiag_phi} and \eqref{Tuv}, we have for $(u,v)\in \cS_{1}\times \cS_{1}$, $u\ne v$,
		\begin{align}
		\bEE{T_{uv}}
		=&\mathbb{E}_{G}\bigg[\frac{1}{n}\bigg(\sum_{j\in I_{1}}\rho \gamma_{1}^{s}+\lambda_{0}^{2}(\gamma_{2}^{s}+\gamma_{3}^{s})
		+\sum_{j\in I_{2}}\rho \gamma_{1}^{d}+\lambda_{0}^{2}(\gamma_{2}^{d}+\gamma_{3}^{d})\nonumber\\
		&~~~~~~~~~~~+\cdots+\sum_{j\in I_{\ell}}\rho \gamma_{1}^{d}+\lambda_{0}^{2}(\gamma_{2}^{d}+\gamma_{3}^{d})\bigg)\bigg]\nonumber \\
		=&\frac{1}{\ell}\bigg(\rho \gamma_{1}^{s}+\lambda_{0}^{2}(\gamma_{2}^{s}+\gamma_{3}^{s})\bigg)
		+\frac{\ell-1}{\ell}\bigg(\rho \gamma_{1}^{d}+\lambda_{0}^{2}(\gamma_{2}^{d}+\gamma_{3}^{d})\bigg)
		\ed\on
		,\label{mu_s}
		\end{align}
		where again we used $\bE{G}{|I_{i}|}=n/\ell$ for all $i\in[\ell]$. This holds for $(u,v)\in \cS_{i}\times \cS_{i}$, for every $i\in[\ell]$.
		
		For the case when $(u,v)\in \cS_{1}\times \cS_{2}$ or when $(u,v)\in \cS_{2}\times \cS_{1}$,
		\begin{align}
		\bEE{T_{uv}}
		=&\mathbb{E}_{G}\bigg[\frac{1}{n}\bigg(\sum_{j\in I_{1}}\rho \gamma_{1}^{sd}+\lambda_{0}^{2}(\gamma_{2}^{sd}+\gamma_{3}^{sd})
		+\sum_{j\in I_{2}}\rho \gamma_{1}^{sd}+\lambda_{0}^{2}(\gamma_{2}^{sd}+\gamma_{3}^{sd})\nonumber\\
		&~~~~~~~+\sum_{j\in I_{3}}\rho \gamma_{1}^{d}+\lambda_{0}^{2}(\gamma_{2}^{d}+\gamma_{3}^{d})+\cdots 
		+\sum_{j\in I_{\ell}}\rho \gamma_{1}^{d}+\lambda_{0}^{2}(\gamma_{2}^{d}+\gamma_{3}^{d})\bigg)\bigg]\\
		=&\frac{2}{\ell}\bigg(\rho \gamma_{1}^{sd}+\lambda_{0}^{2}(\gamma_{2}^{sd}+\gamma_{3}^{sd})\bigg)
		+\frac{\ell-2}{\ell}\bigg(\rho \gamma_{1}^{d}+\lambda_{0}^{2}(\gamma_{2}^{d}+\gamma_{3}^{d})\bigg)
		\ed\off.\label{mu_d}
		\end{align}
		Again, the same expression holds for $\bEE{T_{uv}}$ whenever $(u,v)\in \cS_{i}\times \cS_{j}$, $i,j\in[\ell]$, $i\ne j$.
		The mean matrix $\bEE{T}$ thus has a block structure with $\dg$ on the diagonal, $\on$ on the remaining entries in the diagonal blocks and $\off$ on the off diagonal blocks as depicted in Figure~\ref{fig_block_mtx}. 

\subsection{Proof of Lemma \ref{lem:ET_spectrum}}\label{app:ET_spectrum}

Using the structure of $\bEE{T}$ derived in Lemma \ref{lem:ET_structure}, we have,
		\begin{align*}
		\|\bEE{T}\|_{op}
		&=\dg+(k-1)\on+k(\ell-1)\off\\
		&\le \rho\frac{k^{2}\ell}{m^{2}}+\lambda_{0}^{2}\frac{k^{3}\ell}{m^{2}},
		\end{align*}
		where we have used the definitions in \eqref{M_diag}, \eqref{mu_s} and \eqref{mu_d}, and simplified.
		
	{For the eigengap computation, we first note from the definitions in \eqref{mu_s} and \eqref{mu_d} that}
		\begin{align*}
		\on-\off
		=&\frac{\rho}{\ell}(\gamma_{1}^{s}+\gamma_{1}^{d}-2\gamma_{1}^{sd})
		+\frac{\lambda_{0}^{2}}{\ell}(\gamma_{2}^{s}+\gamma_{2}^{d}-2\gamma_{2}^{sd}+\gamma_{3}^{s}+\gamma_{3}^{d}-2\gamma_{3}^{sd})\\
		=& \frac{\rho}{\ell}\cdot 0
		+\frac{\lambda_{0}^{2}}{\ell}\bigg\{\bigg(1+\frac{c_{2}-1}{m}\bigg)^{2}+\frac{1}{m^{2}}\bigg(2+\frac{c_{2}^{2}-2}{m}\bigg)\\
		&+\frac{1}{m}\bigg(1+\frac{c_{2}-1}{m}\bigg)^{2}-\frac{2}{m}\bigg(1+\frac{c_{2}-1}{m}\bigg)\bigg(1+\frac{2}{m}\bigg)+\frac{4}{m^{2}}\bigg\}\\
		\ge& \frac{\lambda_{0}^{2}}{\ell}.
		\end{align*}
		We therefore have,
		\begin{align*}
	\Delta_{\ell}=	\nu_{\ell}-\nu_{\ell+1}
		=k(\on-\off)\ge \frac{\lambda_{0}^{2}k}{\ell}.
		\end{align*}

	\section{Useful lemmas}
	\label{app:moment_concentration}
\begin{definition}\label{def_sbg}
		A random variable $X$ is subgaussian with variance
	parameter $\sigma^2$, denoted 	$X\sim\sbg(\sigma^2)$, if  
	\begin{equation*}
	\log \bEE{e^{\theta (X-\bEE{X})}}\leq \theta^2\sigma^2/2,
	\end{equation*}
	for all $\theta\in\mR$.
\end{definition}

\begin{definition}\label{def_sbx}
		A random variable $X$ is subexponential with parameters $\sigma^{2}$ and $b>0$, denoted $X\sim \sbx(\sigma^{2},b)$, if
	\begin{equation*}
	\log \bEE{e^{\theta(X-\bEE X)}}\le \theta^{2}\sigma^{2}/2,
	\end{equation*}  
	for all $|\theta|<1/b$.
\end{definition}

	\begin{lemma}\label{lem_sbx_moment}
		Let $X$ be a subexponential random variable with parameters $v^{2}$ and $b>0$, i.e., for every $t>0$,
		\begin{align*}
		\bPr{|X-\bEE{X}|\ge t}
		\le 2\exp\bigg(-\min\bigg\{\frac{t^{2}}{2v^{2}},\frac{t}{2b}\bigg\}\bigg).
		\end{align*}
		Then, for $q\in\mathbb{N}$, and an absolute constant $c$,
		\begin{align*}
		\bEE{|X-\bEE{X}|^{2q}}\le
		2q\cdot(16)^{q}\bigg(\Gamma(q)v^{2q}+b^{2q}\Gamma(2q)\bigg). 
		\end{align*}
	\end{lemma}
		\begin{proof}
			We first express the tail bound for $X$ in a form that is easier to evaluate, and then use standard arguments (see, for example, \cite[Theorem 2.3]{BLM_conc_text}) to derive the moment bound.
			We have,
				\begin{align*}
				\bPr{|X-\bEE{X}|\ge t}
				&\le 2\exp\bigg(-\min\bigg\{\frac{t^{2}}{2v^{2}},\frac{t}{2b}\bigg\}\bigg)\\
				&\le 2\exp\bigg(\frac{-t^{2}}{2(v^{2}+bt)}\bigg),
				\end{align*}
				that is,
				\begin{align*}
				\bPr{|X-\bEE{X}|\ge bu+\sqrt{b^{2}u^{2}+2v^{2}u}}
				\le e^{-u}.
				\end{align*}
				With this tail bound, we can now derive the stated moment bound by using
				\begin{align*}
				\bEE{|X-\bEE{X}|^{2q}}
				=2q\int_{0}^{\infty}\bPr{|X-\bEE{X}|\ge t}t^{2q-1}dt.
				\end{align*}
				In particular, upon substituting $t=bu+\sqrt{b^{2}u^{2}+2v^{2}u}$, we get
				\begin{align*}
				\mathbb{E}\bigg[(X-\bEE{X})^{2q}\bigg]
				&\le  2q\int_{0}^{\infty}e^{-u}(bu+\sqrt{b^{2}u^{2}+2v^{2}u})^{2q-1}\nonumber \\ 
				&~~~~~~~~~~\times \bigg(b+\frac{b^{2}u+v^{2}}{\sqrt{b^{2}u^{2}+2v^{2}u}}\bigg)du,
				\end{align*}
				which after simplification yields
				\begin{align*}
				\mathbb{E}\bigg[(X-\bEE{X})^{2q}\bigg]
				\le  2q\cdot (16)^{q}\bigg(b^{2q}\Gamma(2q)+v^{2q}\Gamma(q)\bigg).
				\end{align*}
				                       
\end{proof}
		\begin{lemma}\label{lem_qform_moment}
			Let $X\in\mR^d$ be a mean zero random vector with independent entries such that $\bEE{X_{i}^{2}}=\lambda_{0}$ and $\bEE{X_{i}^{4}}=\rho$ for all $i\in[d]$. Then, for every $a, b\in\mR^d$,
			\begin{align*}	\bEE{(X^{\top}ab^{\top}X)^{2}}=\rho\sum_{i=1}^{d}a_{i}^{2}b_{i}^{2}+\lambda_{0}^{2}\sum_{i\ne j}(a_{i}^{2}b_{j}^{2}+a_{i}b_{i}a_{j}b_{j}).
			\end{align*} 
			In particular, 
			\begin{align*}
			\bEE{(X^{\top}aa^{\top}X)^{2}}
			=\rho\sum_{i=1}^{d}a_{i}^{4}+2\lambda_{0}^{2}\sum_{i\ne j}a_{i}^{2}a_{j}^{2}.
			\end{align*}
		\end{lemma}
		\begin{remark}
			If the second and fourth moments are related as $\rho=2\lambda_{0}^{2}=2c$ for some absolute constant $c$, then the result simplifies to $\bEE{(X^{\top}ab^{\top}X)^{2}}=c((a^{\top}b)^{2}+\|a\|_{2}^{2}\|b\|_{2}^{2})$.
		\end{remark}
		\begin{proof}
			To start with, we note that the quadratic form $X^{\top}ab^{\top}X$ is a subexponential random variable since $X$ is subgaussian. Although this fact can be used to derive upper bounds on the moments of $X^{\top}ab^{\top}X$, we would like to explicitly compute the second moment. We have,
			\begin{align*}
			\bEE{(X^{\top}ab^{\top}X)^{2}}
			&=\bEE{\bigg(\sum_{i=1}^{d}a_{i}b_{i}X_{i}^{2}+\sum_{i\ne j}a_{i}b_{j}X_{i}X_{j}\bigg)^{2}}\\
			&=\bEE{\bigg(\sum_{i=1}^{d}a_{i}b_{i}X_{i}^{2}\bigg)^{2}
				+\bigg(\sum_{i\ne j}a_{i}b_{j}X_{i}X_{j}\bigg)^{2}+2\sum_{i=1}^{d}a_{i}
				b_{i}X_{i}^{2}\sum_{i\ne j}a_{i}b_{j}X_{i}X_{j}}\\
			&=\bEE{\sum_{i=1}^{d}a_{i}^{2}b_{i}^{2}X_{i}^{4}
				+\sum_{i\ne j}a_{i}b_{i}a_{j}b_{j}X_{i}^{2}X_{j}^{2}+\sum_{i\ne j}a_{i}^{2}b_{j}^{2}X_{i}^{2}X_{j}^{2}}.
			\end{align*} 
			Using $\bEE{X_{i}^{2}}=\lambda_{0}$ and $\bEE{X_{i}^{4}}=\rho$, we get
			\begin{align*}
			\bEE{(X^{\top}ab^{\top}X)^{2}}
			&=\rho\sum_{i=1}^{d}a_{i}^{2}b_{i}^{2}
			+\lambda_{0}^{2}\sum_{i\ne j}(a_{i}^{2}b_{j}^{2}+a_{i}b_{i}a_{j}b_{j}).
			\end{align*}
		\end{proof}
		\begin{lemma}\label{lem_moments_ip}
			Let $X,Y,Z$ and $W$ be independent random vectors taking values in $\mR^{m}$, with independent entries that are zero mean with variance $1/m$. Additionally, for every $i\in[m]$, let $\bEE{Z_{i}^{2q}}=c_{q}/m^{q}$, for q=2, 3, 4 and a constant $c_{q}$ that depends only on $q$. Then, the following results hold:
			\begin{enumerate}[(i)]
				\item $\bEE{\|Z\|_{2}^{4}}=1+\frac{1}{m}(c_{2}-1)$
				
				\item $\bEE{\|Z\|_{2}^{6}}=1+\frac{3}{m}(c_{2}-1)+\frac{1}{m^{2}}(c_{3}-3c_{2}+2)$
				
				\item 
				$\bEE{\|Z\|_{2}^{8}}=1+\frac{6}{\nm}(c_2-1)+\frac{1}{\nm^{2}}(11-18c_2+6c_2^2+4c_3)+\frac{1}{\nm^{3}}(c_4-4c_3-6c_2^2+12c_2-6)$
				
				\item $\bEE{(X^{\top}Y)^{4}}=\frac{2}{m^{2}}+\frac{1}{m^{3}}(c_{2}^{2}-2)$
				
				\item $\bEE{\|Z\|_{2}^{4}(Z^{\top}W)^{2}}
				=\frac{1}{m}\bigg(1+\frac{3}{m}(c_{2}-1)+
				\frac{1}{m^{2}}(c_{3}-3c_{2}+2)\bigg)$
				
				\item $\bEE{(X^{\top}Z)^{2}(X^{\top}W)^{2}}
				=\frac{1}{m^{2}}\bigg(1+\frac{1}{m}(c_{2}-1)\bigg)$
				
				\item $\bEE{\|Z\|_{2}^{2}\|W\|_{2}^{2}(Z^{\top}W)^{2}}
				=\frac{1}{m}\bigg(1+\frac{1}{m}(c_{2}-1)\bigg)^{2}$
				
				\item $\bEE{\|Z\|_{2}^{2}(W^{\top}Z)(X^{\top}Z)(X^{\top}W)}=\frac{1}{m^{2}}\bigg(1+\frac{1}{m}(c_{2}-1)\bigg)$
				
				\item $\bEE{(Z^{\top}X)(Z^{\top}Y)(W^{\top}X)(W^{\top}Y)}=\frac{1}{m^{3}}$
				
				\item $\bEE{(X^{\top}Y)^{2}}=\frac{1}{m}$.
			\end{enumerate}
		\end{lemma}	
		
		\begin{proof}
			\begin{enumerate}[(i)]
				\item 
				\begin{align*}
				\vspace{-5cm}\bEE{\|Z\|_{2}^{4}}
				&=\bEE{\sum_{i=1}^{m}Z_{i}^{4}+\sum_{i\ne j}Z_{i}^{2}Z_{j}^{2}}\\ 
				&=\frac{c_{2}}{m}+\frac{m-1}{m}
				=1+\frac{1}{m}(c_{2}-1).
				\end{align*}

				\item 
				\begin{align*}
				\bEE{\|Z\|^{6}}&=\bEE{(Z_{1}^{2}+\ldots+Z_{\nm}^{2})^{2}(Z_{1}^{2}+\ldots+Z_{\nm}^{2})}\\
				&=\bEE{\bigg(\sum_{i=1}^{\nm}Z_{i}^{4}+\sum_{i\ne j}Z_{i}^{2}Z_{j}^{2}\bigg)\bigg(\sum_{t=1}^{\nm}Z_{t}^{2}\bigg)}\\		&=\bEE{\sum_{i=1}^{\nm}Z_{i}^{4}\sum_{t=1}^{\nm}Z_{t}^{2}+\sum_{t=1}^{\nm}Z_{t}^{2}\sum_{i\ne j}Z_{i}^{2}Z_{j}^{2}}.
				\end{align*}
				For the first term,
				\begin{align*}		\bEE{\sum_{i=1}^{\nm}Z_{i}^{4}\sum_{t=1}^{\nm}Z_{t}^{2}}
				&=\bEE{\sum_{i=1}^{\nm}Z_{i}^{6}+\sum_{i\ne t}Z_{i}^{4}Z_{t}^{2}}\\
				&=\nm\frac{c_3}{\nm^{3}}+\nm(\nm-1)\frac{c_2}{\nm^{2}}\frac{1}{\nm}
				=\frac{1}{\nm^2}(c_3-c_2)+\frac{c_2}{\nm},\numberthis
				\end{align*}
				and for the second term,
				\begin{align*}
				\bEE{\sum_{t=1}^{\nm}Z_{t}^{2}\sum_{i\ne j}Z_{i}^{2}Z_{j}^{2}}
				&=\bEE{2\sum_{t\ne i}Z_{t}^{4}Z_{i}^{2}+\sum_{t\ne i\ne j}Z_{t}^{2}Z_{i}^{2}Z_{j}^{2}}\\
				&= 2m(\nm-1)\frac{c_2}{\nm^{2}}\frac{1}{\nm}+\nm(\nm-1)(\nm-2)\frac{1}{\nm^{3}}\\
				&=1+\frac{1}{\nm}(2c_2-3)-\frac{2}{\nm^{2}}(c_2-1)		
				\end{align*}
				Thus,
				\begin{equation*}
				\bEE{\|Z\|^{6}}=1+\frac{3}{\nm}(c_2-1)+\frac{1}{\nm^{2}}(c_3-3c_2+2).
				\end{equation*}
				
				\item 
				\begin{align*}
				\bEE{\|Z\|^{8}}
				=&\bEE{(Z_{1}^{2}+\cdots+Z_{\nm}^{2})^{4}}\\
				=&\nm\bEE{Z_{1}^{8}}+{\nm\choose 2} \frac{4!}{3!}2\bEE{ Z_{1}^{6}Z_{2}^{2}}+{\nm\choose 2} \frac{4!}{2!2!}2\bEE{ Z_{1}^{4}Z_{2}^{4}}
				+{\nm\choose 3} \frac{4!}{2!}3\bEE{ Z_{1}^{4}Z_{2}^{2}Z_{3}^{2}}\nonumber\\
				&+{\nm\choose 4} 4!\bEE{Z_{1}^{2}Z_{2}^{2}Z_{3}^{2}Z_{4}^{2}}\\
				=&1+\frac{6}{\nm}(c_2-1)+\frac{1}{\nm^{2}}(11-18c_2+6c_2^2+4c_3)+\frac{1}{\nm^{3}}(c_4-4c_3-6c_2^2+12c_2-6).
				\end{align*}
				
				\item 
				To compute $\bEE{(X^{\top}Y)^{4}}$, we first note that
				\begin{align*}
				\bEE{(X^{\top}Y)^{4}|X}&=\bEE{(Y^{\top}XX^{\top}Y)^{2}|X}\\
				&=\bEE{Y_{1}^{4}}\sum_{i=1}^{m}X_{i}^{4}+2(\bEE{Y_{1}^{2}})^{2}
				\sum_{i\ne j}X_{i}^{2}X_{j}^{2}\\
				&=\frac{c_2}{m^{2}}\sum_{i=1}^{m}X_{i}^{4}+2\bigg(\frac{1}{m}\bigg)^{2}
				\sum_{i\ne j}X_{i}^{2}X_{j}^{2},
				\end{align*}
				where we used Lemma \ref{lem_qform_moment} in the second step. This gives
				\begin{align*}
				\bEE{(X^{\top}Y)^{4}}&=\frac{c_{2}}{m}\bEE{X_{1}^{4}}+\frac{2(m-1)}{m}(\bEE{X_{1}^{2}})^{2}\\
				&=\frac{c_{2}^{2}}{m^{3}}+\frac{2(m-1)}{m^{3}}=\frac{2}{m^{2}}+\frac{1}{m^{3}}(c_{2}^{2}-2).
				\end{align*}
				
				\item 
				Similar to the previous calculation, we first compute the conditional expectation to get
				\begin{align*}
				\bEE{\|Z\|_{2}^{4}(Z^{\top}W)^{2}|Z}
				&=\|Z\|_{2}^{4}\bigg(\sum_{i=1}^{m}\bEE{Z_{i}^{2}W_{i}^{2}|Z}+\sum_{i\ne j}\bEE{Z_{i}W_{i}Z_{j}W_{j}|Z}\bigg)=\|Z\|_{2}^{4}\frac{\|Z\|_{2}^{2}}{m},
				\end{align*}
				which gives
				\begin{align*}
				\bEE{\|Z\|_{2}^{4}(Z^{\top}W)^{2}}&=\frac{1}{m}\bEE{\|Z\|_{2}^{6}}=\frac{1}{m}\bigg(1+\frac{3}{\nm}(c_2-1)+\frac{1}{\nm^{2}}(c_3-3c_2+2)\bigg).
				\end{align*}
				
				\item 
				We have 
				\begin{align*}
				\bEE{(X^{\top}Z)^{2}(X^{\top}W)^{2}|X}
				&=\bEE{(X^{\top}Z)^{2}|X}\bEE{(X^{\top}W)^{2}|X}=\frac{\|X\|_{2}^{2}}{m}\cdot\frac{\|X\|_{2}^{2}}{m}.
				\end{align*}
				Thus,
				\begin{align*}
				\bEE{(X^{\top}Z)^{2}(X^{\top}W)^{2}}
				=\frac{1}{m^{2}}\bigg(1+\frac{1}{m}(c_{2}-1)\bigg).
				\end{align*}

				\item 
				\begin{align*}
				\bEE{\|Z\|_{2}^{2}\|W\|_{2}^{2}(Z^{\top}W)^{2}|Z}
				=&\|Z\|_{2}^{2}~\bEE{\|W\|_{2}^{2}(Z^{\top}W)^{2}|Z}\\
				=&\|Z\|_{2}^{2}\bigg(\sum_{i=1}^{m}\bEE{\|W\|_{2}^{2}Z_{i}^{2}W_{i}^{2}|Z}+\sum_{i\ne j}\bEE{\|W\|_{2}^{2}W_{i}W_{j}Z_{i}Z_{j}|Z}\bigg)\\
				=&\|Z\|_{2}^{2}\sum_{i=1}^{m}Z_{i}^{2}\bEE{W_{i}^{4}+\sum_{l\ne i}W_{i}^{2}W_{l}^{2}}\\
				&+\|Z\|_{2}^{2}\sum_{i\ne j}Z_{i}Z_{j}\bEE{W_{i}^{3}W_{j}+W_{j}^{3}W_{i}+\sum_{l\ne i,~l\ne j}W_{l}^{2}W_{i}W_{j}}\\
				=&\|Z\|_{2}^{2}\sum_{i=1}^{m}Z_{i}^{2}\bigg(\frac{c_{2}}{m^{2}}+\frac{m-1}{m^{2}}\bigg)=\|Z\|_{2}^{4}\bigg(\frac{1}{m}+\frac{c_{2}-1}{m^{2}}\bigg).
				\end{align*}
				Thus,
				\begin{align*}
				\bEE{\|Z\|_{2}^{2}\|W\|_{2}^{2}(Z^{\top}W)^{2}}=\frac{1}{m}\bigg(1+\frac{c_{2}-1}{m}\bigg)^{2}.
				\end{align*}
				
				\item 
				\begin{align*}
				\bEE{\|Z\|_{2}^{2}(W^{\top}Z)(X^{\top}Z)(X^{\top}W)|Z,W}
				&=\|Z\|_{2}^{2}(W^{\top}Z)\bEE{X^{\top}WZ^{\top}X|W,Z}\\
				&=\|Z\|_{2}^{2}(W^{\top}Z)\frac{Z^{\top}W}{m}.
				\end{align*}
				Using similar arguments as in the proof of (v),
				\begin{align*}
				\bEE{\|Z\|_{2}^{2}(W^{\top}Z)(X^{\top}Z)(X^{\top}W)}=\frac{1}{m^{2}}\bigg(1+\frac{c_{2}-1}{m}\bigg).
				\end{align*}

				\item 
				\begin{align*}
				\bEE{(Z^{\top}X)(Z^{\top}Y)(W^{\top}X)(W^{\top}Y)|X,Y,W}
				&=(W^{\top}X)(W^{\top}Y)\bEE{Z^{\top}XY^{\top}Z|X,Y}\\
				&=(W^{\top}X)(W^{\top}Y)\frac{X^{\top}Y}{m}
				\end{align*}
				Thus,
				\begin{align*}
				\bEE{(Z^{\top}X)(Z^{\top}Y)(W^{\top}X)(W^{\top}Y)}
				&=\frac{1}{m}\bE{X,Y}{\bE{W}{(W^{\top}X)(W^{\top}Y)(X^{\top}Y)|X,Y}}\\
				&=\frac{1}{m}\bE{X,Y}{(X^{\top}Y)\bE{W}{W^{\top}XY^{\top}W|X,Y}}\\
				&=\frac{1}{m^{2}}\bE{X,Y}{(X^{\top}Y)^{2}}=\frac{1}{m^{3}}.
				\end{align*}

				\item \begin{align*}
				\bEE{(X^{\top}Y)^{2}}&=\sum_{i=1}^{m}\bEE{X_{i}^{2}Y_{i}^{2}}+\sum_{i\ne j}\bEE{X_{i}Y_{i}X_{j}Y_{j}}=\frac{1}{m}.
				\end{align*}
			\end{enumerate}
		\end{proof}	

\end{appendices}

\bibliography{IEEEabrv,bibfile}
\bibliographystyle{IEEEtran}

\end{document}